\documentclass[twoside,11pt]{article}
\usepackage{jair}

\usepackage{hyperref}
\usepackage[sort]{natbib}
\usepackage{comment}
\usepackage{bbm}
\usepackage{nicefrac}
\usepackage{xspace}
\usepackage{xcolor}
\usepackage{amsthm,amssymb,amsmath}
\usepackage{mathrsfs}
\usepackage{mathtools}
\usepackage[labelformat=simple]{subcaption}

\usepackage{enumitem}
\usepackage[toc,page]{appendix}
\usepackage{appendix}
\usepackage{tabularx}
\newcolumntype{Y}{>{\centering\arraybackslash}X}
\usepackage{booktabs} 
\renewcommand\thesubfigure{(\alph{subfigure})}

\usepackage{tikz}
\usetikzlibrary{shapes,arrows.meta,calc,fit,backgrounds,shapes.multipart,positioning}
\tikzset{
    ib/.style={
    rectangle split, rectangle split parts=2,
    rectangle split draw splits=false,
    rounded corners=\blockroundedcorners,
    line width=\blocklinewidth,
    minimum width=\IBwidth,
    text width=\IBwidth-\blocklinewidth*2,
    draw=#1,
    rectangle split part fill={#1, none},
    text=white, text centered
    },
    every two node part/.style={align=left, text=black}
}
\newlength\IBheaderheight 
\newlength\IBwidth 
\newlength\IBheight 
\newlength\blockroundedcorners
\newlength\blocklinewidth
\newlength\NodeDistanceW 
\newlength\NodeDistanceH 
\definecolor{mycolor}{RGB}{100,100,100}
\definecolor{myteal}{RGB}{0,190,200}

\newcommand\ImageNode[4][]{
  \node[ib={mycolor},#1,rounded corners=0.1cm] (#2)
  {{\footnotesize#3}\nodepart{two}\centering\includegraphics[width=\IBwidth-0.1cm]{#4}};
}

\newcommand{\bi}{\begin{itemize}}
\newcommand{\ei}{\end{itemize}}
\newcommand{\be}{\begin{enumerate}}
\newcommand{\ee}{\end{enumerate}}
\newcommand{\bb}{\begin{block}}
\newcommand{\eb}{\end{block}}

\newcounter{row}
\makeatletter
\@addtoreset{subfigure}{row}
\renewcommand\p@subfigure{\thefigure.}
\makeatother

\newlength{\tempdima}
\newcommand{\rowname}[1]{
    \rotatebox{90}{\makebox[\tempdima][c]{#1}}
}

\def\gameSubfigWidth{0.33\textwidth}
\newcommand{\gameRow}[3]{
    \settoheight{\tempdima}{\includegraphics[width=\gameSubfigWidth]{figs/#1_trajs.png}}%
    \rowname{#2}&
    \subcaptionbox{\label{fig:#3_trajs}}{%
        \centering%
        \includegraphics[width=\gameSubfigWidth]{figs/#1_trajs.png}}&
    \subcaptionbox{\label{fig:#3_crossings}}{%
        \centering%
        \includegraphics[width=\gameSubfigWidth]{figs/#1_crossings.png}%
    }&
    \subcaptionbox{\label{fig:#3_avg_opponent}}{%
        \centering%
        \includegraphics[width=\gameSubfigWidth]{figs/#1_avg_opponent.pdf}%
    }\\
    \stepcounter{row}%
}

\usepackage[capitalise]{cleveref}  
\crefname{equation}{}{}
\newtheorem{theorem}{Theorem}[section]
\newtheorem{proposition}[theorem]{Proposition}

\newtheorem{definition}[theorem]{Definition}
\newtheorem{lemma}[theorem]{Lemma}
\newtheorem{remark}[theorem]{Remark}

\makeatletter
\renewcommand\p@subfigure{\thefigure}
\makeatother

\title{Evolutionary Dynamics and\\$\Phi$-Regret Minimization in Games}

\author{\name Georgios Piliouras$^{1}$ \email georgios.piliouras@gmail.com \\
       \name Mark Rowland$^{2}$ \email markrowland@deepmind.com \\
       \name Shayegan Omidshafiei$^{2}$ \email somidshafiei@deepmind.com \\
       \name Romuald Elie$^{2}$ \email relie@deepmind.com \\
       \name Daniel Hennes$^{2}$ \email hennes@deepmind.com \\
       \name Jerome Connor$^{2}$ \email jeromeconnor@deepmind.com \\
       \name Karl Tuyls$^{2}$\email karltuyls@deepmind.com \\
       \addr $^{1}$SUTD\\
       \addr $^{2}$DeepMind\\
       }

\begin{document}

\maketitle
\begin{abstract}
Regret has been established as a foundational concept in online learning, and likewise has important applications in the analysis of learning dynamics in games.
Regret quantifies the difference between a learner's performance against a baseline in hindsight.
It is well-known that regret-minimizing algorithms converge to certain classes of equilibria in games;
however, traditional forms of regret used in game theory predominantly consider baselines that permit deviations to deterministic actions or strategies. 
In this paper, we revisit our understanding of regret from the perspective of deviations over partitions of the full \emph{mixed} strategy space (i.e., probability distributions over pure strategies), under the lens of the previously-established $\Phi$-regret framework, which provides a continuum of stronger regret measures.
Importantly, $\Phi$-regret enables learning agents to consider deviations from and to mixed strategies, generalizing several existing notions of regret such as external, internal, and swap regret, and thus broadening the insights gained from regret-based analysis of learning algorithms. 
We prove here that the well-studied evolutionary learning algorithm of replicator dynamics (RD) seamlessly minimizes the strongest possible form of $\Phi$-regret in generic $2 \times 2$ games, without any modification of the underlying algorithm itself. 
We subsequently conduct experiments validating our theoretical results in a suite of 144 $2 \times 2$ games wherein RD exhibits a diverse set of behaviors.
We conclude by providing empirical evidence of $\Phi$-regret minimization by RD in some larger games, hinting at further opportunity for $\Phi$-regret based study of such algorithms from both a theoretical and empirical perspective.
\end{abstract}

\section{Introduction}\label{sec:intro}
Understanding the behavior of learning dynamics in games is a fundamental problem studied in game theory, online learning theory, dynamical systems, and multiagent systems. 
Numerous works have been developed in the area that describe connections between these distinct fields \citep{Tuyls03,Tuyls06,TuylsP07,cesa2006prediction,nisan_agt_book_2007,Chang07,WunderLB10,klos2010evolutionary,Galstyan13,BloembergenTHK15,GattiR16,SrinivasanLZPTM18,CelliMF020,Vlatakis-Gkaragkounis20}.
Given this wealth of prior work, a natural question emerges: 
have we converged to a more-or-less complete mathematical language that allows us to accurately describe the behavior of learning dynamics (at least in simple small games), or are there remaining characteristics of these dynamics that are not yet well understood or fully described by existing concepts?

Arguably, one of the most important concepts in this area is that of \emph{regret} \citep{Auer95,Freund99,Chang05,Kleinberg09multiplicativeupdates,Roughgarden09,monnot2017limits}.
Regret is a basic definition of online learning with numerous variants, each with specific associated properties. 
Regret minimization forms the basis of a key class of algorithms for learning in games, and is concurrently central to the more general field of online sequential prediction \citep{cesa2006prediction,Chang07}.
Informally, regret is defined as the difference in the cumulative performance of an agent against a baseline, which typically allows hindsight deviations (or swaps) of one or more of the agent's actions with alternative actions.
There exist various notions of regret, ranging from basic forms such as external regret~\citep{hannan1957approximation}, which measures a learner's performance against the best fixed action in hindsight, to stronger variants such as internal regret~\citep{foster1998asymptotic} and swap-regret~\citep{blum2007external}, where the deviating action of the agent is allowed to be a function of their originally chosen action.
These notions of regret have typically been defined within the context of single-agent online learning, and subsequently used to analyze behavior of decision-making in multiagent games.
The benefit of regret-minimizing algorithms, when applied to games, is that they typically yield time-average convergence to a corresponding set of equilibria.
For example, in general-sum games, (external) regret minimization by all players yields time-average convergence to coarse correlated equilibria~\citep{nisan_agt_book_2007};
in two-player zero-sum games, this yields convergence to Nash equilibria, and has driven key successes towards solving games such as Poker~\citep{zinkevich2007regret,sandholm2010state,moravvcik2017deepstack,brown2018superhuman,2020_rebel_brown}.

Despite these insights, traditional notions of regret such as those described above are, at times, too general to provide strong guarantees for certain learning dynamics in games.
Even in zero-sum games, for example, standard regret-minimizing dynamics such as multiplicative weights update~\citep{arora2012multiplicative} and its well-known continuous-time limit counterpart, the replicator dynamics (RD), can be non-convergent or even chaotic in their  real-time behaviors~\citep{bailey2018multiplicative,cheung2019vortices,cheung2020chaos,sato2002chaos,piliouras2014optimization}, despite their time-average convergence.
Emergence of chaos can even materialize in simple congestion games~\citep{palaiopanos2017multiplicative,chotibut2020route}, and perhaps surprisingly, such behaviors can occur despite these algorithms' time-average convergence to equilibria~\citep{Freund99}.
In fact, such prototypical learning algorithms can be non-equilibrating even in the smallest of environments, $2\times2$  games with two agents and two strategies, e.g., Matching Pennies~\citep{papadimitriou2016nash,bailey2019fast,chotibut20family}. 
Overall, the behavior of well-known dynamics even in small games is diverse and non-trivial, and clearly cannot be fully inferred from standard regret analysis alone.

The clear gap between traditional regret theory and the observed empirical behaviors of these well-studied algorithms hints at the enticing possibility of using stronger notions of regret to better understand dynamical behaviors at finer levels of granularity.
As it is well-known that regret-minimizing algorithms must, in general, randomize (i.e., output probability distributions over actions, rather than deterministic actions), such a stronger notion could consider a regret definition where agents may condition their deviating behavior not merely on deterministic actions (e.g., as in traditional notions of regret), but using more refined deviation functions.
In investigating this, a line of prior work has analyzed a stronger concept known as $\Phi$-regret~\citep{greenwald2003general,gordon2008no,stoltz2007learning,hongpractical}, which encapsulates several more general classes of deviations in contrast to traditional regret notions.
However, as later detailed, while some of these works have introduced algorithms for $\Phi$-regret minimization, such algorithms have either been investigated only for specific and simple classes of $\Phi$-regret, or apply to more general classes of games albeit being significantly more intricate (e.g., require more book-keeping and are more difficult to implement).

In this paper, our primary contribution is to establish a link between $\Phi$-regret and the simple and well-studied evolutionary dynamics algorithm of RD.
Our key theoretical result is that in general classes of $2 \times 2$ games, RD in self-play seamlessly minimizes the strongest possible notion of $\Phi$-regret, without any modification of the underlying learning dynamics.
 Informally, this result implies that an RD learner, in self-play, attains on time-average at least the value of the game, even if time-averaging is applied to specific recurrent parts of the trajectory as defined by the agent's mixed strategy.
We further ground these theoretical results in empirical analysis focusing on two-player games, introducing a version of $\Phi$-regret we denote `mosaic regret', which is more amenable to empirical implementation.
We illustrate mosaic regret-convergence under RD in a broad range of 144 $2 \times 2$ games \citep{bruns2015names}, with widely varying characteristics (e.g., fully cooperative games, social dilemmas, cyclical games, etc.).
Finally, we conclude by showing empirical evidence that RD minimizes mosaic regret in some larger games, hinting at future avenues of further exploration.

The remainder of the paper is structured as follows. 
In \cref{sec:prelim}, we overview the necessary preliminaries for establishing our theoretical results.
In \cref{sec:trad_to_phi_regret}, we delve into the $\Phi$-regret framework and related concepts targeted in the paper.
Following this, we establish our theoretical results in \cref{sec:analysis,sec:linking_swap_other}, and subsequently validate them empirically in \cref{sec:experiments}.
Finally, we conclude with key discussion points and takeaways in \cref{sec:discussion}.

\section{Preliminaries}\label{sec:prelim}
We first review preliminaries related to game theory and online learning algorithms.

\subsection{Game Theory}\label{sec:prelim_gt}
We study two-player normal-form games, where the first (resp., second) player has access to a finite set of \emph{pure} strategies $\mathcal{A}^1$ (resp., $\mathcal{A}^2$).
In two-player games, players $1$ and $2$ are also referred to as the row and column players, respectively.
The joint \emph{mixed} strategies of the players are denoted by $(x,y)$, where $x \in \Delta(\mathcal{A}^1)$ and $y \in \Delta(\mathcal{A}^2)$. We will denote the probability that the first (resp., second player) assigns to their $i$-th strategy (resp., $j$-th strategy) as $x_i$ (resp., $y_j$).
The player payoffs are, respectively, specified by matrices $A \in \mathbb{R}^{n \times m}$ and $B \in \mathbb{R}^{n \times m}$, where $n$ and $m$ are the number of strategies available to each player.
Let $a_{ij}, b_{ij}$ represent the payoff entries of the respective matrices.
The respective utilities received by the players are $x^{\top}Ay$ and $x^{\top}By$.
In a zero-sum (resp., coordination) game, players receive payoffs $A = -B$ (resp., $A=B$).
Given strategy profile $(x,y)$, the best response for each player is the strategy that maximizes their utility against the other player's current strategy.
The \emph{best response dynamics} arise when players iteratively update their policy to their best response.

In game theory, the Nash equilibrium has been well-established as a solution concept of interest, and is defined as follows.
\begin{definition}
A mixed strategy profile $(x_*,y_*)\in\Delta(\mathcal{A}^1)\times\Delta(\mathcal{A}^2)$ is a Nash equilibrium (NE) if
\begin{equation}\label{eq:nash}
    x_{*}^{T}Ay_* \geq x^{\top}Ay_* \quad \forall x \in \Delta(\mathcal{A}^1) \quad \text{ and } \quad 
    x_{*}^{T}By_* \geq x_{*}^{T}By \quad \forall y \in \Delta(\mathcal{A}^2) \, .
\end{equation}
\end{definition}
In other words, the players are simultaneously in best response with one another when their profiles constitute an NE. 

Several weaker notions of equilibria are also important in game theory, and bear a close relationship to regret-minimizing strategies. In preparation for those definitions, it will be useful to introduce some notation that allows for us to work with arbitrary correlated probability distributions over $\mathcal{A}=\mathcal{A}^1\times \mathcal{A}^2$. Let $z \in \Delta(\mathcal{A})$ be such a probability distribution. 
Let $z_{ij}$ be the probability assigned to the outcome where the first (resp., second) player choose strategy $i$ (resp., $j$).
Let $z(1|\cdot)$ (resp., $z(2|\cdot)$) be the marginal probability of the first (resp., second) player.
We denote by $z(2|i) \in \Delta(\mathcal{A}^2)$  the conditional distribution of the second player's strategy given the first player's strategy is $i$, and similarly define $z(1|j)$. Finally, we denote by $e_{i} \in \Delta(\mathcal{A}^1)$ the distribution putting mass 1 on $i \in \mathcal{A}^1$, and similarly define $e_{j}$. Given these definitions, we next introduce several additional notions of equilibria.
\begin{definition}
A distribution over joint strategies $z \in\Delta(\mathcal{A}^1 \times \mathcal{A}^2)$ is a correlated equilibrium (CE) if
\begin{align}\label{eq:ce}
     (e_{i})^\top A z(2|i) &\geq x^\top A z(2|i) \quad \forall i \in \mathcal{A}^1 \, , x \in \Delta(\mathcal{A}^1) \text{ and } \\
     z(1|j)^\top B e_{j} &\geq z(1|j)^\top B y \quad \forall j \in \mathcal{A}^2 \, , y \in \Delta(\mathcal{A}^2) \, .
\end{align}
\end{definition}

\begin{definition}
A distribution over joint strategies $z \in\Delta(\mathcal{A}^1 \times \mathcal{A}^2)$ is a coarse correlated equilibrium (CCE) if
\begin{align}\label{eq:cce}
     \sum_{i,j} a_{ij}z_{ij} &\geq x^\top A z(2|\cdot) \quad \forall x \in \Delta(\mathcal{A}^1) \text{ and } \\
     \sum_{i,j} b_{ij}z_{ij} &\geq z(1|\cdot)^\top B y \quad \forall  y \in \Delta(\mathcal{A}^2) \, . 
\end{align}
\end{definition}

It is important to note that the set of CCE is a superset of the CE, which itself is a superset of NE (i.e., $\text{NE}\subset \text{CE} \subset \text{CCE}$). 
An intuitive way to establish this set of inclusions is to define each of these solution concepts via allowable sets of joint distributions and classes of allowable deviating policies under which no player can strictly improve their payoff. 
From this perspective, a NE is a \textit{product} of (mixed) strategies such that no player can deviate to another (mixed) strategy and strictly increase their payoff. 
The set of CCE can be defined equivalently merely by removing the restriction that the joint distribution of the two players has to be a product of each of the player's marginal distributions. Hence, CCE is a superset of NE. Finally, if the players' initial joint distribution is not a mixture, one can define strictly more powerful deviating strategies by allowing a player's deviating strategy to depend on their realized strategy.
This is the case for CE, where the set of allowable deviating strategies is strictly more expansive, and yet none of them can can strictly improve the payoffs of the players; 
thus, the set of CE is more restrictive than CCE, while generalizing NE as, once again, it allows for correlated joint distributions.

\subsection{Replicator Dynamics}
In this paper, we seek to investigate the connection of more expansive notions of regret to simple, well-studied learning algorithms.
For these purposes, we focus on the replicator dynamics (RD), a standard and well-studied model  defining the evolution of strategic, interacting individuals under biologically-inspired mechanisms \citep{taylor1978evolutionary,schuster1983replicator}.
In the two-player setting of interest, the time-evolution of player strategies is described by RD as follows,
\begin{align}
    \dot{x}_i = x_i\left((Ay)_i-x^{\top}Ay\right) \quad \dot{y}_j = y_j\left((x^{\top}B)_j-x^{\top}By\right) \quad \forall (i,j) \in \mathcal
    {A}^1 \times \mathcal{A}^2 \,.
\end{align}
RD is the continuous-time variant of the well-known multiplicative weights update (MWU) meta-algorithm~\citep{arora2012multiplicative,Kleinberg09multiplicativeupdates}, and the seminal dynamics in the areas of mathematical evolution, biology, ecology, and evolutionary game theory ~\citep{Weibull,Hofbauer98}.
In recent years, RD has enjoyed a particularly strong surge in applications to learning in multiplayer games~\citep{GallaFarmer_ScientificReport18,papadimitriou2019game,omidshafiei2019alpha,boone2019darwin,lanctot2019openspiel,nagarajan2020chaos,flokas2020no,hennes2020neural,sorin2020replicator,skoulakis2021evolutionary}.
Despite its algorithmic simplicity, RD is well-known to minimize external regret (a concept later detailed in \cref{sec:external_and_swap_regret}), thus yielding time-average convergence to a coarse correlated equilibrium~\citep{sorin2009exponential,mertikopoulos2018cycles}.

\subsection{Online Sequential Prediction}
The central problem of online sequential prediction in continuous time is specified by the interaction of a player with finite action set $\mathcal{A}$ and an environment at each time $t \in [0, \infty)$. 
At time $t$, the player selects a distribution over actions $x^t \in \Delta(\mathcal{A})$, and a utility function $u^t \in \mathbb{R}^{\mathcal{A}}$ is revealed.\footnote{The typical language of online learning uses loss vectors, $\ell^t = -u^t$. Here, we use utilities to match the focus on payoffs, rather than losses, which is common in the game theory literature.} 
The player's expected instantaneous utility is $\langle x^t, u^t \rangle$, and the player is shown the entire utility vector $u^t$, which it may use in deciding how to act in subsequent interactions. After interacting up to time $T>0$, the cumulative utility experienced by the player is $\int_{0}^T \langle x^t, u^t \rangle \mathrm{d}t$. It is difficult to judge how well the player has done in selecting its actions on the basis of this cumulative utility alone; the quality of the performance depends on whether there were other actions available that would have yielded significantly higher utility. 
This is formalized by judging the player's performance based on its \emph{regret}, a common means of studying the performance of such algorithms in both discrete and continuous time~\citep{zinkevich2003online,cesa2006prediction,blum2007learning,shoham2008multiagent,kwon2017continuous,harris1998rate,mertikopoulos2018cycles,sorin2009exponential,banerjee2004performance}.

\section{From Traditional Regret Minimization to $\Phi$-Regret and Mosaic Regret}\label{sec:trad_to_phi_regret}
We next lay the foundations for our theoretical results, by overviewing a spectrum of no-regret algorithms: from the more traditional forms of regret to the more general $\Phi$-regret framework, and our introduced notion of mosaic regret.

\subsection{External and Swap Regret}\label{sec:external_and_swap_regret}
\begin{figure}[t]
    \centering
    \begin{subfigure}{0.4\textwidth}
        \centering
        \includegraphics[width=0.9\textwidth,page=1]{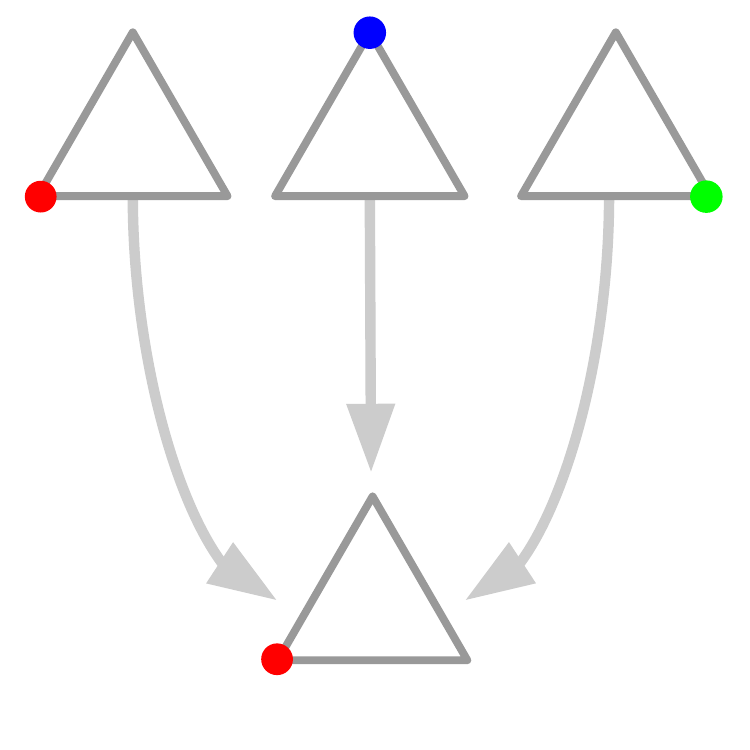}%
        \caption{Example external regret deviations.}
        \label{fig:regret_comparisons_external}
    \end{subfigure}%
    \hspace{25pt}
    \begin{subfigure}{0.4\textwidth}
        \centering
        \includegraphics[width=0.9\textwidth,page=2]{figs/regret_comparisons.pdf}%
        \caption{Example swap regret deviations.}
        \label{fig:regret_comparisons_swap}
    \end{subfigure}
    \hfill
    \caption{Strategy deviations associated with different forms of regret.
    This example considers a simple online learning setting with three actions (i.e., pure strategies) available to the player, illustrated via the simplices in each row.
    \subref{fig:regret_comparisons_external} External regret considers deviations of all three possible actions to a single, fixed, deterministic action in hindsight. \subref{fig:regret_comparisons_swap} Swap regret considers a function independently swapping \emph{each} of the player's actions with an alternative fixed action.
    }
    \label{fig:regret_comparisons}
\end{figure}

Informally, regret quantifies whether a player could have done better by using an alternative method for picking actions throughout the interaction.
One of the most common variants studied is \emph{external regret}, the expected improvement in performance that could have been achieved by sticking with a single action throughout all interactions during the time interval $[0,T]$.
\Cref{fig:regret_comparisons_external} provides an illustrative example of the action deviations considered when computing external regret.
This example considers a simple online learning setting with three actions available to the player (illustrated via each of the 2-simplices), with external regret considering deviations of all three actions to a single, fixed action in hindsight.
External regret is mathematically defined by
\begin{align}\label{eq:external-regret}
    \max_{a \in \mathcal{A}} \int_0^T u^t_a \mathrm{d}t - 
    \int_{0}^T \langle x^t, u^t \rangle \mathrm{d}t\, .
\end{align}
If a player is able to attain $o(T)$ external regret, over bounded utility sequences ($u^t \in [0,1]^\mathcal{A}$, for example), then the player's strategy is said to minimize external regret, or is simply \emph{regret-minimizing}. Intuitively, the sub-optimality of the player's decision at each timestep, relative to the best constant action in hindsight, becomes vanishingly small, no matter what sequence of utilities $(u^t)_{t \in [0,T]}$ are yielded by the environment.

\paragraph{Regret in games.}
Regret minimization is central to many algorithms for computing equilibria in game theory, due to the close relationship between the notion of an alternative action, and the game-theoretic notions of strategy deviations that feature in the definitions of equilibria in \cref{sec:prelim_gt}.
As such, one can consider \emph{actions} in the sense of online learning as being synonymous with (pure) \emph{strategies} in game theory; we henceforth use the latter terminology for simplicity.
For instance, consider a two-player game as specified in \cref{sec:prelim_gt}.
Let us cast the problem the players face in playing the game as an online sequential prediction problem, focusing on the first player in the following description.
The first player's strategy set is $[n]$, the set of strategies in the game. 
The utility vector $u^t$ for this player at time $t$ is given by $Ay^t$, where $y^t$ is the strategy selected by the second player at time $t$.

A well-known Folk-theorem implies that if both players employ algorithms that minimize external regret to select their strategies, the players' time-average behavior (i.e., the joint strategy $T^{-1}\int_0^T (x^t, y^t) \mathrm{d}t$) is guaranteed to converge to the set of coarse correlated equilibrium~\citep{hart2000simple,young2004strategic,roughgarden2016twenty}.
Further, if the game is zero-sum, then the product of the marginals of their individual time-averaged strategies converges to the set of Nash equilibria at the same rate~\citep{Freund99,nisan_agt_book_2007,young2004strategic}. A similar relationship holds between the set of correlated equilibria, and the stronger notions of regret described below~\citep{hart2000simple}.

\paragraph{Broader deviation classes.}
There exist more general notions of regret that compare a player's behavior against a wider class of baselines than just those that deviate to a single, fixed strategy throughout time.
One such alternative, \emph{swap regret}~\citep{blum2007external}, permits deviations involving the player using strategy $b\in\mathcal{A}$ every time they had selected strategy a $a\in\mathcal{A}$.
\Cref{fig:regret_comparisons_swap} illustrates swap regret in our earlier example, where now each of the three possible pure strategies may be independently deviated to a different one.
The notion of swap regret is formalized through \emph{swap} functions $F : \mathcal{A} \rightarrow \mathcal{A}$ that can be lifted to a function $\bar{F} : \Delta(\mathcal{A}) \rightarrow \Delta(\mathcal{A})$
\begin{align*}
    \bar{F}(x)_{b} = \sum_{\substack{a \in \mathcal{A} \\ F(a) = b}} x_a \, , 
\end{align*}
for all $x \in \Delta(\mathcal{A})$.
The swap regret is then defined to be
\begin{align*}
    \max_{F: \mathcal{A} \rightarrow \mathcal{A}} \int_0^T \langle \bar{F}(x^t), u^t \rangle \mathrm{d}t - \int_{0}^T \langle x^t, u^t \rangle \mathrm{d}t  \, .
\end{align*}

Players interacting in a two-player game using algorithms that minimize swap regret are guaranteed to converge to the set of correlated equilibria in time-average, a stronger notion that the coarse correlated equilibrium guaranteed by algorithms that minimize external regret. A slightly weaker notion than swap regret is that of \emph{internal regret}, which restricts the swap functions $F: \mathcal{A} \rightarrow \mathcal{A}$ that are lifted to deviations on the simplex to take the form $F(a) = a$ for all but one $a \in \mathcal{A}$.  
However, note that in the special case of zero-sum games, internal and swap regret minimization algorithms do not offer stronger guarantees, as external regret minimization algorithms already guarantee convergence to a Nash equilibrium in the time-average.

\subsection{$\Phi$-regret Framework}

\begin{figure}[t]
    \centering
    \begin{subfigure}{0.3\textwidth}
        \begin{tikzpicture}[align=center,node distance = \NodeDistanceH and \NodeDistanceW, auto,
            myarrow/.style={-{Latex},line width=\blocklinewidth,mycolor},
            myarrowhighlight/.style={-{Latex},line width=\blocklinewidth,myteal},
            myarrowr/.style={-{Latex},line width=\blocklinewidth,mycolor,rounded corners=5pt},
            myarrowhighlightr/.style={-{Latex},line width=\blocklinewidth,myteal,rounded corners=5pt},
        ]
            \ImageNode{input}{Input}{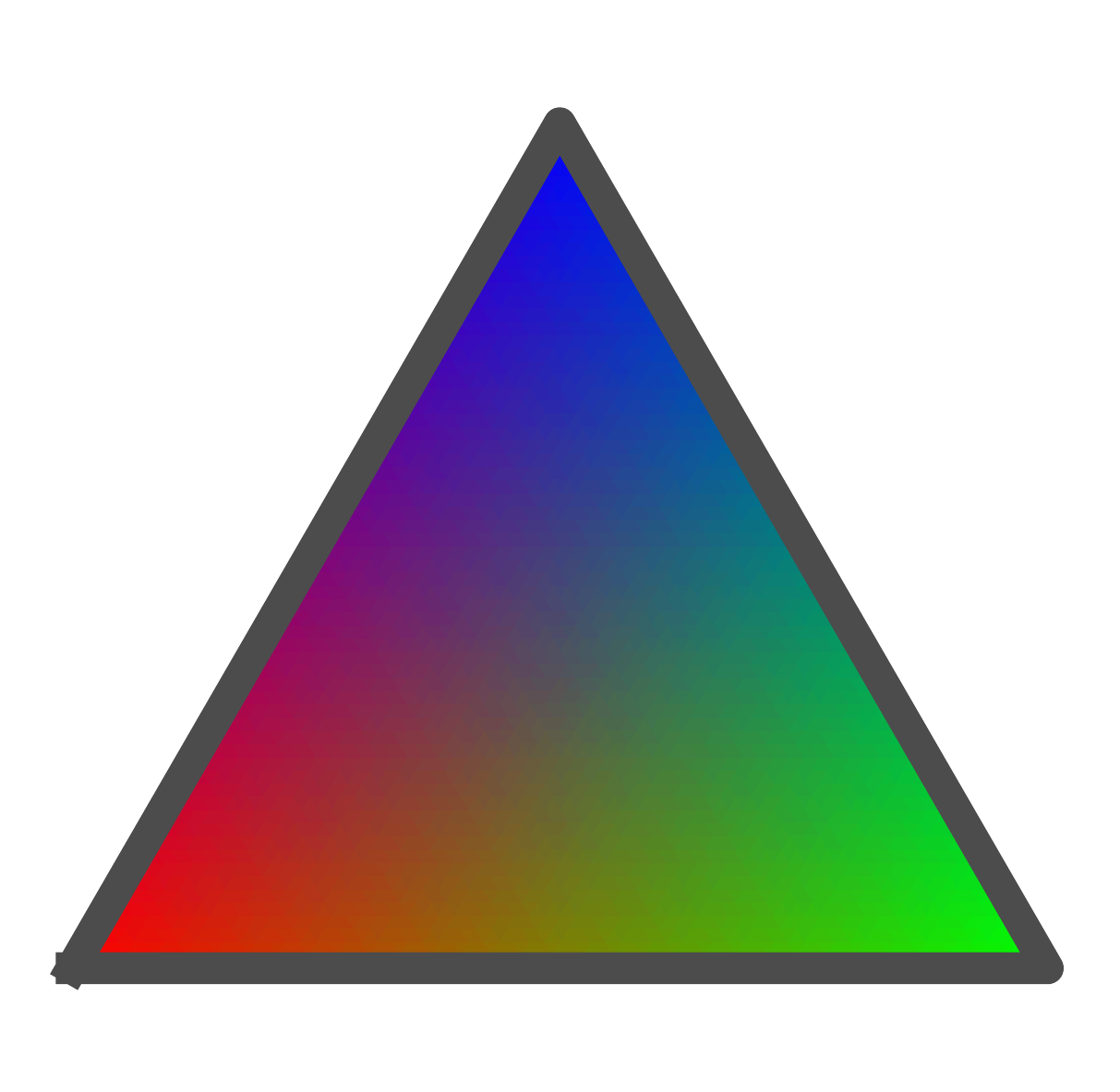}
            \ImageNode[below left=of input]{dev1}{\#1}{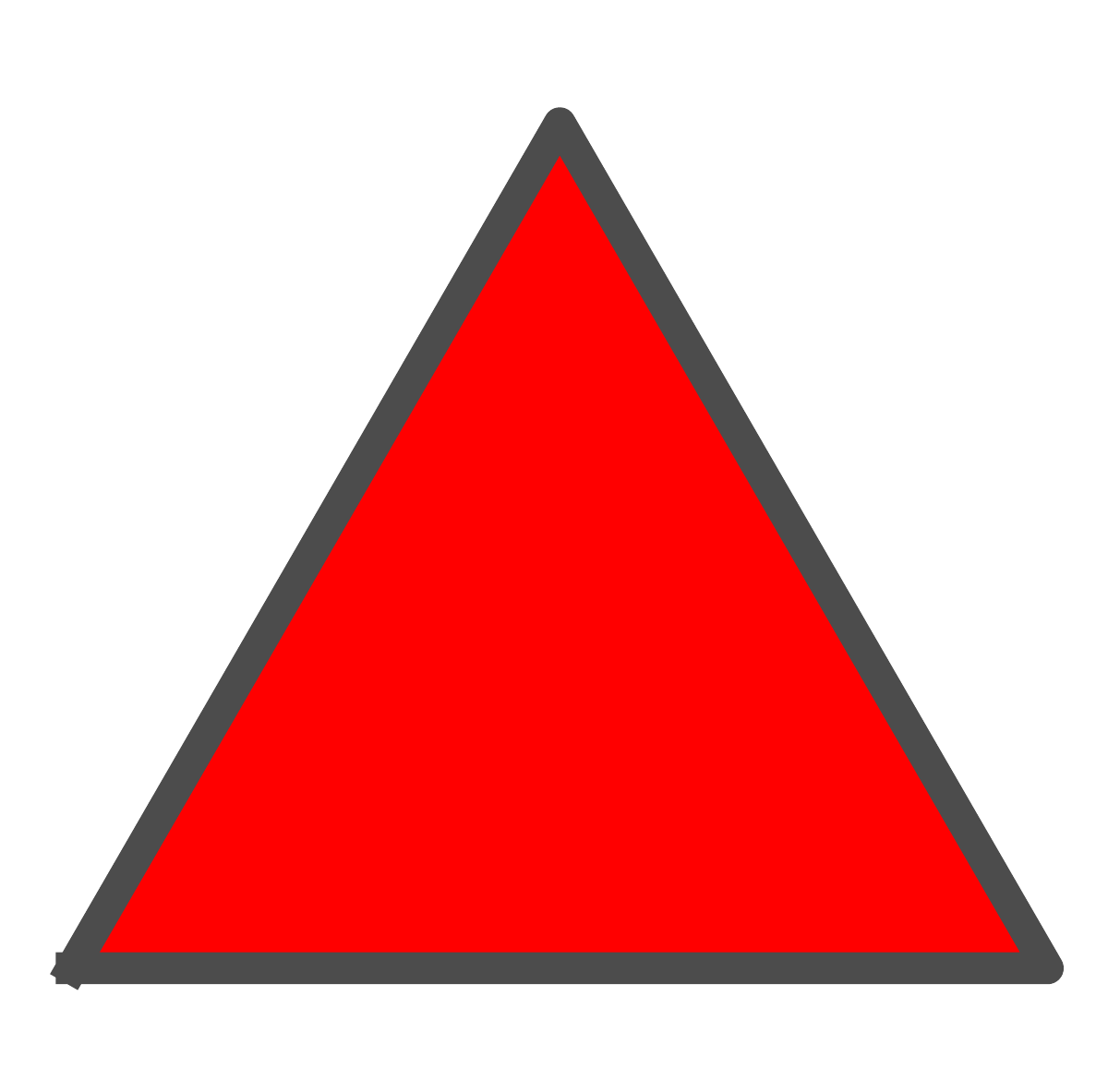}
            \ImageNode[below=of input]{dev2}{\#2}{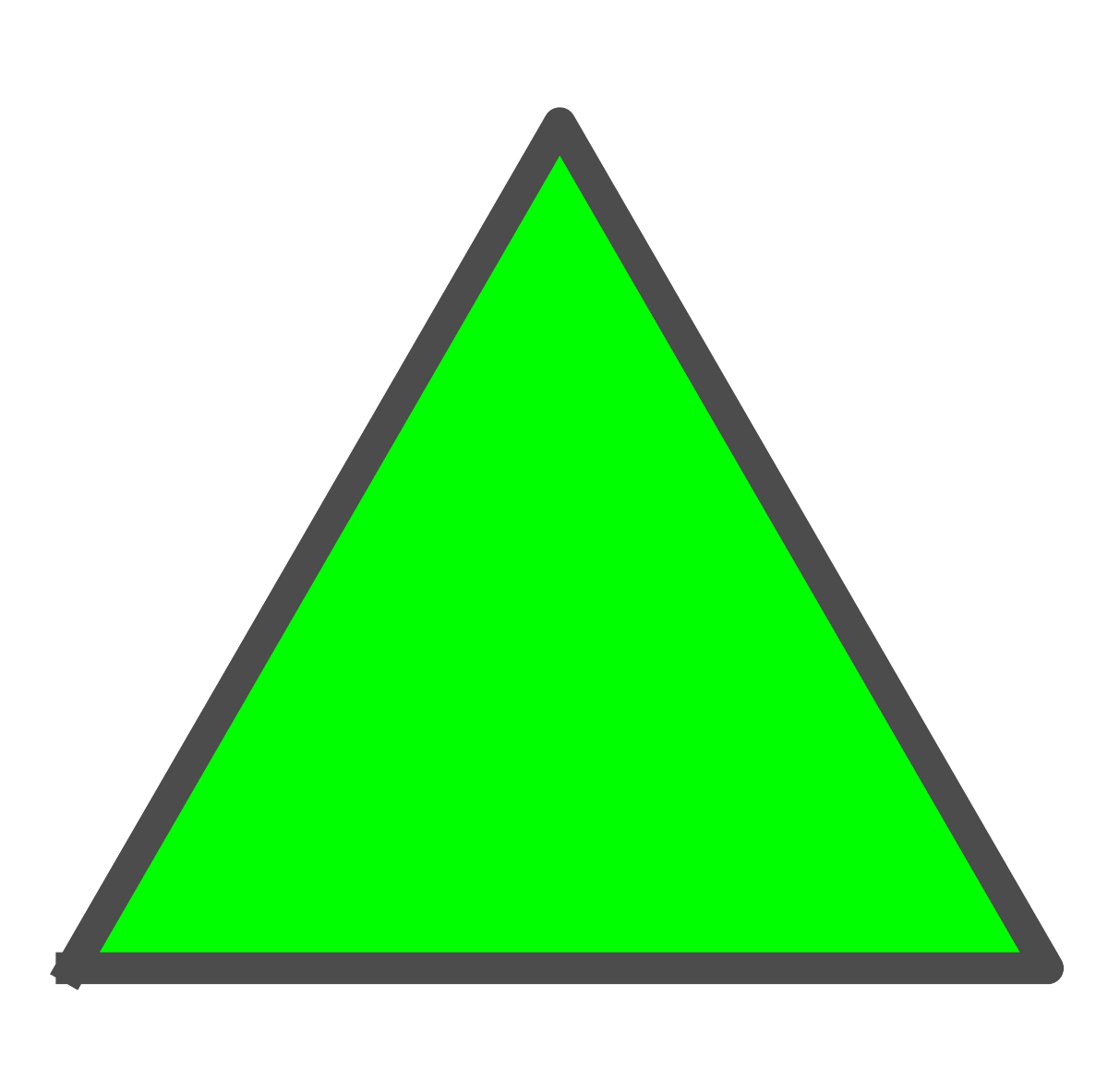}
            \ImageNode[below right=of input]{dev3}{\#3}{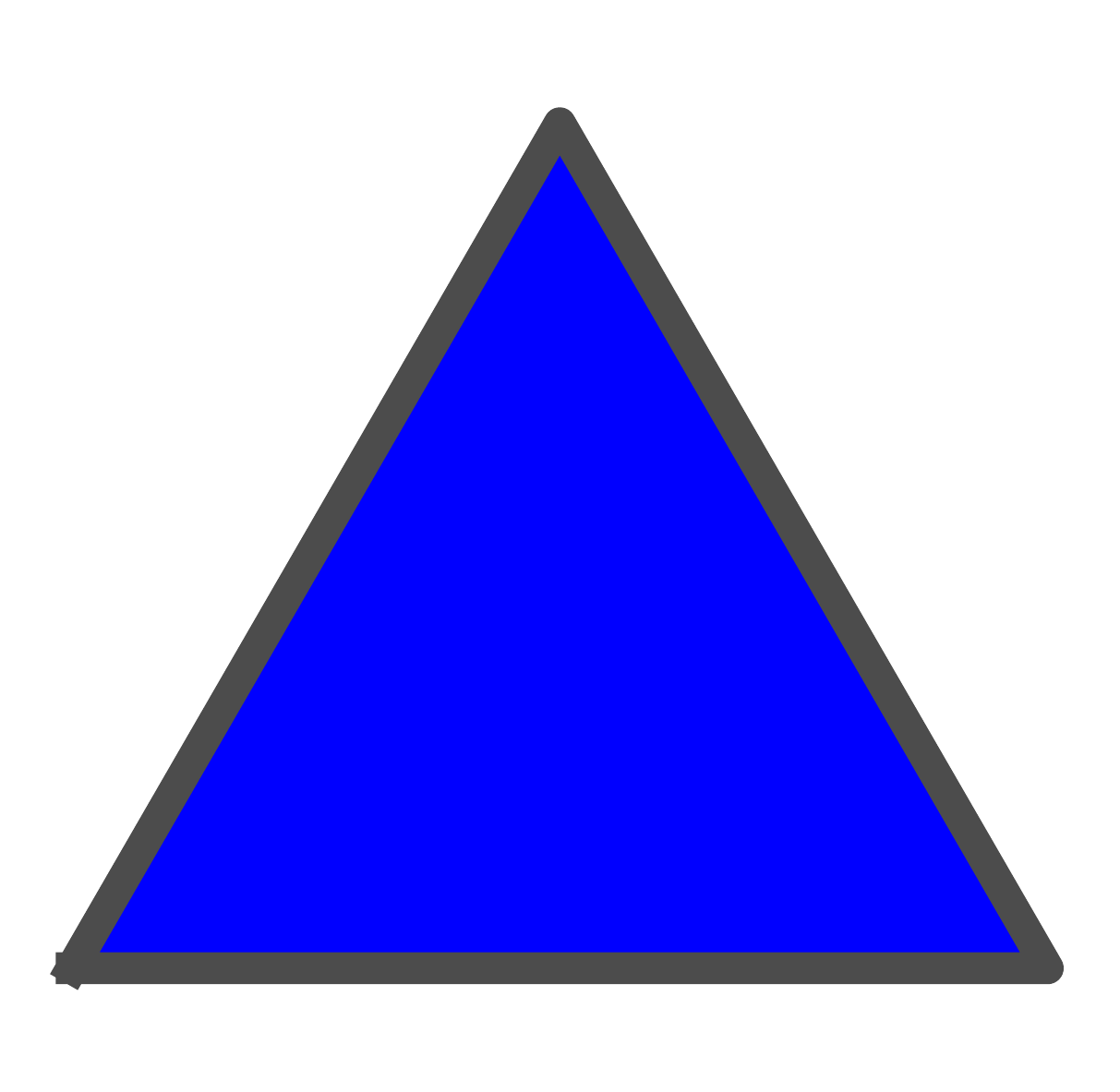}
            \draw [myarrowr] (input.south) |- (dev2.north);
            \draw [myarrowr] (input.south) -- + (0,-0.5) -| (dev1.north);
            \draw [myarrowr] (input.south) -- + (0,-0.5) -| (dev3.north);
        \end{tikzpicture}
        \caption{External Regret}
        \label{fig:alt_external_reg}
    \end{subfigure}%
    \hfill
    \begin{subfigure}{0.3\textwidth}
        \centering
        \begin{tikzpicture}[align=center,node distance = \NodeDistanceH and \NodeDistanceW, auto,
            myarrow/.style={-{Latex},line width=\blocklinewidth,mycolor},
            myarrowhighlight/.style={-{Latex},line width=\blocklinewidth,myteal},
            myarrowr/.style={-{Latex},line width=\blocklinewidth,mycolor,rounded corners=5pt},
            myarrowhighlightr/.style={-{Latex},line width=\blocklinewidth,myteal,rounded corners=5pt},
        ]
            \ImageNode{input}{Input}{figs/dev_input.pdf}
            \ImageNode[below left=of input]{dev1}{\#1}{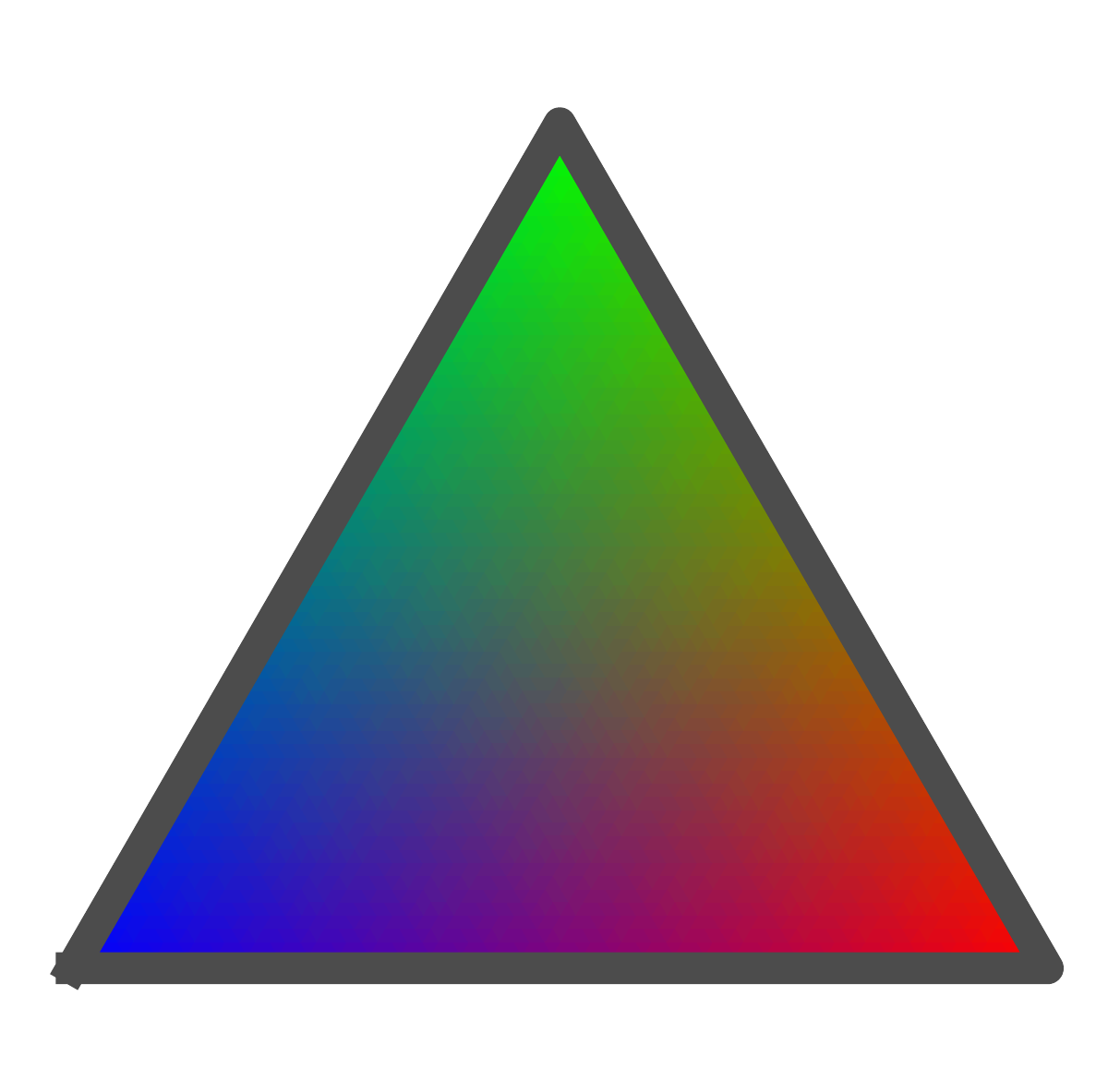}
            \ImageNode[below=of input]{dev2}{\#2}{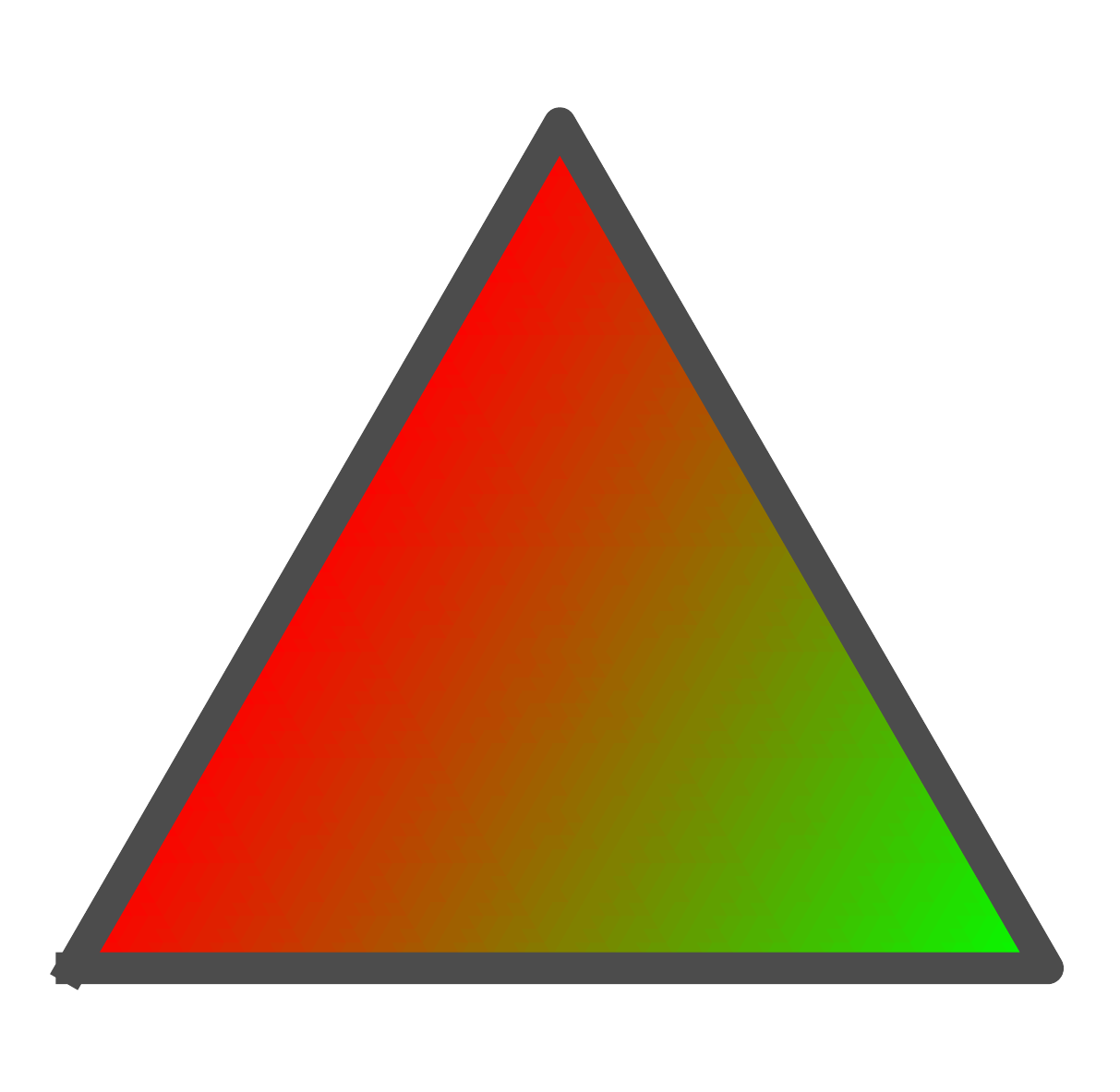}
            \ImageNode[below right=of input]{dev3}{\#3}{figs/dev_external_2.pdf} 
            \draw [myarrowr] (input.south) |- (dev2.north);
            \draw [myarrowr] (input.south) -- + (0,-0.5) -| (dev1.north);
            \draw [myarrowr] (input.south) -- + (0,-0.5) -| (dev3.north);
        \end{tikzpicture}
        \caption{Swap Regret}
        \label{fig:alt_swap_reg}
    \end{subfigure}%
    \hfill
    \begin{subfigure}{0.3\textwidth}
        \centering
        \begin{tikzpicture}[align=center,node distance = \NodeDistanceH and \NodeDistanceW, auto,
            myarrow/.style={-{Latex},line width=\blocklinewidth,mycolor},
            myarrowhighlight/.style={-{Latex},line width=\blocklinewidth,myteal},
            myarrowr/.style={-{Latex},line width=\blocklinewidth,mycolor,rounded corners=5pt},
            myarrowhighlightr/.style={-{Latex},line width=\blocklinewidth,myteal,rounded corners=5pt},
        ]
            \ImageNode{input}{Input}{figs/dev_input.pdf}
            \ImageNode[below left=of input]{dev1}{\#1}{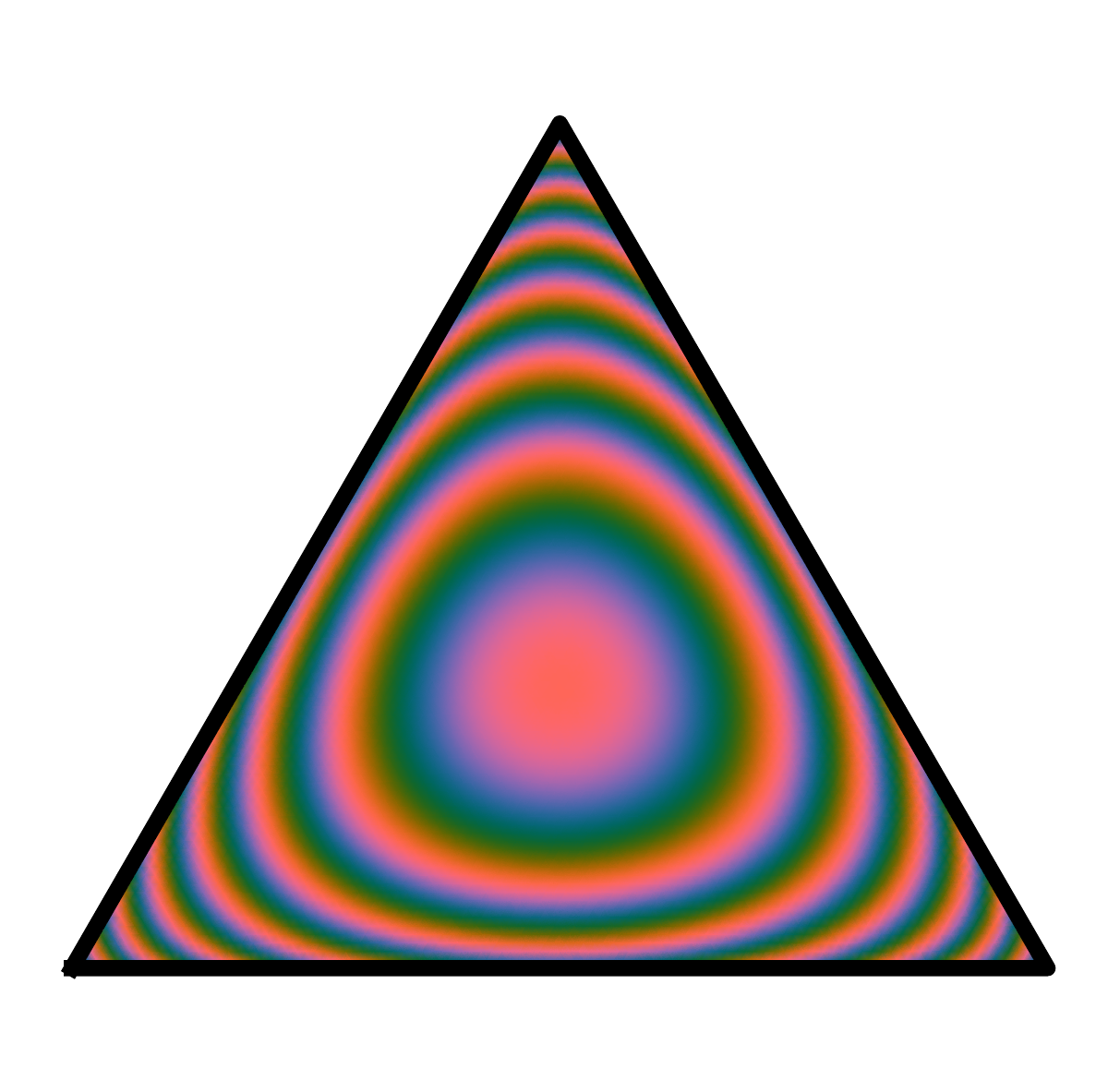}
            \ImageNode[below=of input]{dev2}{\#2}{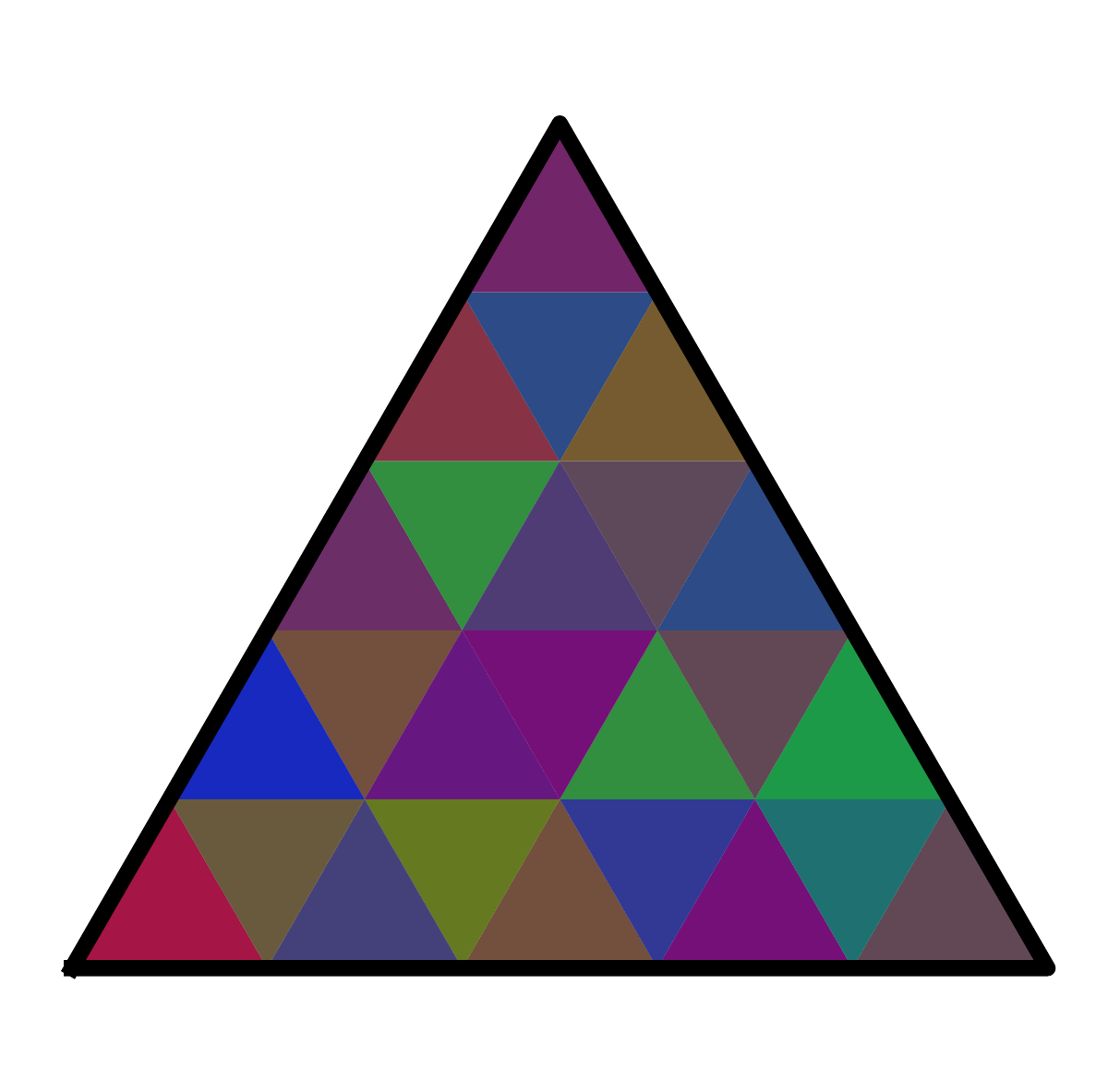}
            \ImageNode[below right=of input]{dev3}{\#3}{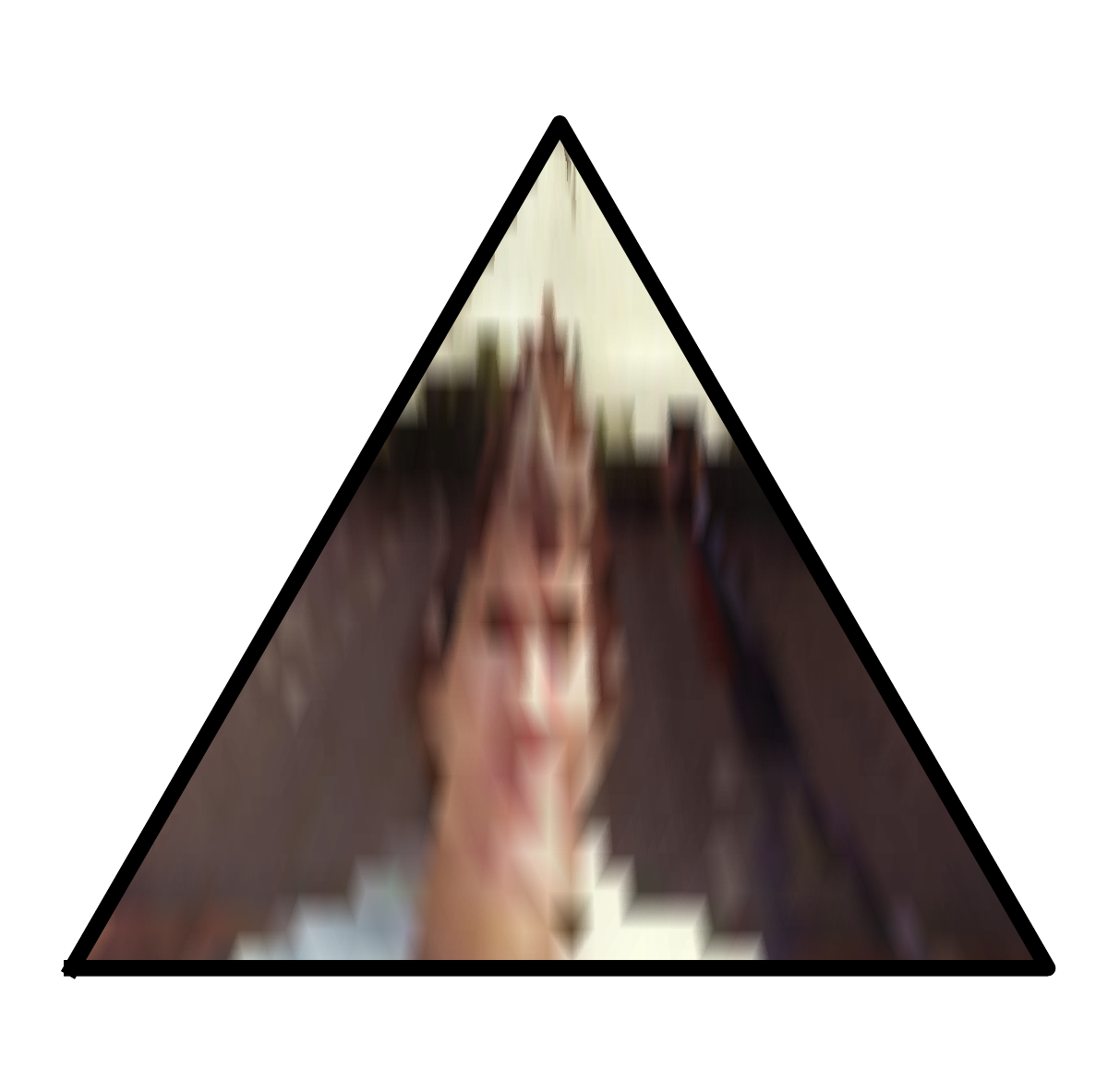}
            \draw [myarrowr] (input.south) |- (dev2.north);
            \draw [myarrowr] (input.south) -- + (0,-0.5) -| (dev1.north);
            \draw [myarrowr] (input.south) -- + (0,-0.5) -| (dev3.north);
        \end{tikzpicture}
        \caption{$\Phi$-regret}
        \label{fig:alt_our_reg}
    \end{subfigure}
    \caption{
    Mixed-strategy view of permissible deviations under various forms of regret. 
    We visualize here a simple 3-strategy game, where the color of each mixed strategy in the simplex interior corresponds to the combination of pure strategies (red, green, and blue) that compose it.
    \subref{fig:alt_external_reg} Under external regret, permitted deviations map all mixed strategies to either the red, green, or blue deterministic strategy. 
    \subref{fig:alt_swap_reg} Swap regret deviations can correspond to affine transformations of the input strategies, generalizing external regret.
    \subref{fig:alt_our_reg} Under $\Phi$-regret, the input space may be partitioned in a multitude of ways.
    For instance, in example \#1, the simplex is discretized into partitions, each capturing mixed strategies with a certain range of KL-divergences from the center of the simplex. 
    Likewise, in example \#2, a discrete set of deviation partitions is defined in a granular manner enabling the output mappings visualized.
    This definition allows for arbitrarily complex mappings, as illustrated by example \#3, which reproduces an image from the three primary strategies (colors).
    }
    \label{fig:alternative_view}
\end{figure}

In swap regret, each swap function $F: \mathcal{A} \rightarrow \mathcal{A}$ is lifted to a function $\bar{F} : \Delta(\mathcal{A}) \rightarrow \Delta(\mathcal{A})$, representing an admissible deviation in the algorithm's past strategies. Notably, each such $\bar{F}$ is \emph{affine} in the input strategy. This is also true of the deviation functions $\bar{F} : \Delta(\mathcal{A}) \rightarrow \Delta(\mathcal{A})$ permitted under external and internal regret. A natural question to pose is whether there exist learning algorithms that minimize regret against even broader classes of deviation functions.

To better intuit the possibilities, let us consider an alternative view of the deviations permitted by various regret notions in \cref{fig:alternative_view}.
As in our earlier example, we consider a simple 3-strategy game in this instance.
However, we now also visualize the full spectrum of mixed strategies, with the color of each mixed strategy in the interior of the simplex corresponding to its composition of the three underlying pure strategies (red, green, and blue).
For instance, \cref{fig:alt_external_reg} visualizes the three possibilities of deviations permitted by external regret, each respectively mapping all mixed strategies to either the red, green, or blue pure strategy. 
Likewise, in the swap regret case shown in \cref{fig:alt_swap_reg}, deviations can correspond to affine transformations of the input strategies.
Here, the deterministic swap function $F$ yields a mixed swap function $\bar{F}$, mapping each input mixed strategy to a distinct output mixed strategy. 
Naturally, as swap regret generalizes external regret, deviations permitted in the latter are also possible under the swap regret case (e.g., as shown in example deviation \#3 in \cref{fig:alt_swap_reg}).

\paragraph{$\Phi$-regret.} A previous line of work \citep{greenwald2003general,gordon2008no,stoltz2007learning} introduces the concept of $\Phi$-regret, expanding the space of deviations permitted compared to traditional regret notions.\footnote{$\Phi$-regret as introduced by \citet{greenwald2003general} applied only to linear deviation functions over $\Delta(\mathcal{A})$. Later work \citep{gordon2008no,stoltz2007learning} applied the concept much more generally to convex games on compact action spaces. The definition we use in this paper is a specific case of this notion, applied to the mixed extension of normal-form games.} 
For our purposes, $\Phi$-regret is defined by first specifying a set of deviation functions $\Phi \subseteq \{ \bar{F} : \Delta(\mathcal{A}) \rightarrow \Delta(\mathcal{A}) \}$, and then defining the corresponding $\Phi$-regret as
\begin{align*}
    \max_{\phi \in \Phi} \int_0^T \langle \phi(x^t), u^t \rangle \mathrm{d}t  - \int_0^T \langle x^t, u^t \rangle \mathrm{d}t \, .
\end{align*}
An algorithm is said to minimize $\Phi$-regret if the above expression is $o(T)$; that is, it grows sub-linearly with time.

\Cref{fig:alt_our_reg} visualizes example deviations permitted under $\Phi$-regret, illustrating that the input space may be partitioned in a multitude of ways under the above definition.
For instance, in example \#1 of \cref{fig:alt_our_reg}, partitions are defined as a function of KL-divergence from the center of the simplex. 
Likewise, in example \#2 of \cref{fig:alt_our_reg}, deviations are defined in a manner enabling the highly granular output mappings visualized.
Theoretically, this definition would allow for any arbitrarily complex mapping, as illustrated by example \#3 of \cref{fig:alt_our_reg}, which reproduces an image from the three primary strategies (colors).

Mathematically, the $\Phi$-regret framework generalizes the concepts of external regret, internal regret, and swap regret. 
Additional specific instances of $\Phi$-regret, however, can be considerably stronger than these notions of regret, as any continuous function on $\Delta(\mathcal{A})$ can be uniformly approximated by piecewise affine functions.

\paragraph{Strong Swap Regret.}
The strongest version of $\Phi$-regret takes $\Phi$ to be as large as possible; the set of all measurable functions from $\Delta(\mathcal{A})$ to itself.
We refer to this notion as \emph{strong swap regret}, as it is defined explicitly as a mapping over the mixed-strategy simplex and may, thus, be \emph{non-affine} (in contrast to the traditional notion of swap regret, which is first defined over pure strategies before being lifted to the mixed space).
Past work has led to the development of algorithms that are able to minimize this notion of regret, which is an extremely strong property, although the algorithms necessarily to do so are rather intricate and impractical to implement~\citep{stoltz2007learning}.
One of the central findings of our work is that in certain circumstances, the straightforward algorithm of replicator dynamics actually minimizes strong swap regret.

\subsection{Mosaic Regret: A Practical Strong Swap Regret Variant}\label{sec:mosaic_regret}
In moving from theory to experiments, strong swap regret can no longer practically be minimized due to time discretization effectively preventing the learning dynamics from revisiting specific mixed strategies. 
It is, therefore, useful to consider a slight weakening of this notion to enable practical evaluation, which we define as follows.\footnote{For related reasons, \citet{gordon2008no} introduce a similarly-weakened notion known as finite-element $\Phi$-regret, which restricts deviations to be continuous, and affine on each piece of a finite partition. Our introduced notion of mosaic regret is a slight strengthening of finite-element $\Phi$-regret, allowing for non-continuous deviations.}
\begin{definition}[Mosaic regret]
    A learning algorithm that produces a strategy process $(x^t)_{t \geq 0}$ in response to a utility process $(u^t)_{t \geq 0}$ is said to minimize mosaic regret if, for any finite partition $\Omega = \{\Omega_1,\ldots,\Omega_K\}$ of $\Delta(\mathcal{A})$, the algorithm has vanishing regret relative to all deviation strategies that are affine on the elements of the partition $\Omega$. Mathematically, let $D(\Omega) = \{ G : \Delta(\mathcal{A}) \rightarrow \Delta(\mathcal{A}) \mid G \text{ affine on each } \Omega_k \}$. Then the algorithm minimizes mosaic regret if
    \begin{align*}
        \max_{G \in D(\Omega)} \int_0^T \langle G(x^t), u^t \rangle \mathrm{d}t  - \int_0^T \langle x^t, u^t \rangle \mathrm{d}t = o(T)
    \end{align*}
    over measurable, bounded utilities $u \in ([0,1]^\mathcal{A})^{[0, \infty)}$, for each finite partition $\Omega$ of $\Delta(\mathcal{A})$.
\end{definition}
Mosaic regret also encompasses many standard existing notions of regret. 
Minimization of mosaic regret, for example, encompasses minimization of external regret through the singleton partition $\Omega = \{\Delta(\mathcal{A})\}$ and constant deviation maps. Internal and swap regret are also encompassed through the singleton partition, and various classes of affine maps. 
We shall return to this introduced notion of mosaic regret in our later experiments in \cref{sec:experiments}.

\section{Analysis of Replicator Dynamics under $\Phi$-Regret}\label{sec:analysis}

We next present theoretical results pertaining to the dynamics of games wherein players use RD.
We focus our attention on two-player $2 \times 2$ normal-form games, which, despite their size, form an important class capturing numerous canonical games traditionally used for closed-form analysis of regret-minimizing algorithms \citep{guyer19722,klos2010evolutionary,bruns2015names,robinson2005topology,rapoport1966taxonomy}. 
We first begin by deriving preliminary results, subsequently establishing that RD minimizes strong swap regret in all classes of $2 \times 2$ games.
Finally, we conclude by revisiting strong swap regret in the context of online learning, and in comparison to the standard notion of external regret. 

We seek to establish the key result that strong swap regret is minimized in all generic $2\times2$ games when both players are using RD. 
The notion of genericity that we consider here is that we only allow for games where, for each pure strategy of their opponent, each player has a unique (pure) best response as well as that for both payoff matrices $A,B$  it must hold that  $a_{11}-a_{12}-a_{21}+a_{22}\neq0$ and $b_{11}-b_{12}-b_{21}+b_{22}\neq0$.\footnote{A natural sufficient condition for the uniqueness of best response strategies is that all payoff entries are distinct.
It is easy to see that in $2\times2$ games under uniqueness of best responses, other useful generic properties follow, such as a finite number of possible Nash equilibria. 
Indeed, if, given a pure strategy of the opponent, each strategy has a unique best response then there cannot exist any mixed Nash equilibrium where exactly one of the two agents is randomizing. So the only possibility for having a continuum of Nash equilibria is in its interior. However, given any two interior Nash equilibria,
any linear combination of them is also a Nash equilibrium, where in fact both agents are strictly indifferent between their two strategies. This line of equilibria intersects the boundary, implying the existence of a strategy profile where at least one of the agents plays a pure strategy and their opponent is indifferent between their two strategies, reaching a contradiction.} Clearly, these properties are satisfied by all but a zero-measure set of games from the space of all $2\times 2$ games. 

To prove the minimization of strong swap regret as outlined above, we consider the three possible classes of such games: (I) those with no pure Nash equilibrium (that have a unique interior, i.e., fully mixed, Nash equilibrium), (II) those with a unique pure Nash equilibrium, and (III) those with (at least)\footnote{Due to our genericity assumptions we can only have at most two pure Nash equilibria.}
two pure Nash equilibria on the boundary.
From the perspective of a \emph{best response graph} that connects pure strategy profiles via directed edges from one to another by a single player best-responding, these cases respectively correspond to: (I) the best response graph is a cycle through the four pure strategy outcomes, (II) the best response graph has a single sink that is a pure (strict) Nash equilibrium, and (III) the best response graph has two sinks/pure (strict) Nash equilibria. 
In case (III), due to our genericity assumption, both players choose a different pure strategy in each pure Nash equilibrium, which implies the existence of an isolated fully mixed Nash equilibrium as well.   

\subsection{Equivalent and Rescaled Games}\label{sec:equiv_rescaled_games}

Before proving the result that RD minimizes strong swap regret in generic $2 \times 2$ games,  we will show in this section that in order to prove case (I) it is sufficient to prove this result for the case of rescaled zero-sum games with a unique interior Nash equilibrium, whereas to prove case (III) it is sufficient to consider the case of rescaled coordination games with a unique interior Nash equilibrium.
Let us first define these notions of equivalent and rescaled games below.

\begin{definition}[Equivalent games up to column/row shifts]
    Let $G=(A,B)$ be a two player game with payoff matrices $A$ and $B$. We say that a game $G'=(A', B')$ is equivalent to $G$  up to column/row shifts if there exists $(c_1,c_2,d_1,d_2)$ such that 
    \begin{align*}
        a'_{ij}= a_{ij}+c_j, \quad b'_{ij}= b_{ij}+d_i \quad \text{for all}~ (i,j)\in\{1,2\}\;.
    \end{align*}
    We then write $(A,B)\sim(A',B')$ to denote equivalence up to column/row shifts.
\end{definition}

A game is called \textit{rescaled zero-sum} (resp., \textit{rescaled coordination game}) if there exist a payoff matrix $C$ and a negative (resp., positive) constant $c$ such that $(A,B)\sim(C,cC)$. 
Observe that since the RD update expression for the row player can be written as
\begin{align*}
    \dot{x}_i=x_i \sum_j x_j \left((Ay)_i-(Ay)_j\right)  \, ,
\end{align*}
the row player's RD trajectories only depend on payoff differences, which remain invariant after column payoff shifts. 
Hence, shifting a game to an equivalent version does not affect RD, implying the following remark.
\begin{remark}\label{remark:orbit}
    The orbits in RD of two equivalent games are identical. 
\end{remark}

The following simple lemma (see also \citet{Hofbauer98}) establishes that every generic $2 \times 2$ game is RD trajectory equivalent to either a rescaled zero-sum game or a rescaled coordination game.

\begin{lemma}\label{lemma:equivalence}
\label{lemma:equi}
    Every generic $2 \times 2$ game is either a rescaled zero-sum game or a rescaled coordination game.
\end{lemma}

\begin{proof}
By assumption we have that
$b_{11}-b_{12}-b_{21}+b_{22}$ and  $a_{11}-a_{12}-a_{21}+a_{22}$ be different from $0$.
Then, there exists $c\neq0$ such that 
$b_{11}-b_{12}-b_{21}+b_{22}= c (a_{11}-a_{12}-a_{21}+a_{22})$.
Let $D:=cA-B$ then $d_{11}-d_{12}-d_{21}+d_{22}=0$. Thus, for all $i,j$, we have $d_{ij}= d_{i2}+(d_{2j}-d_{22})=u_{i}+v_{j}$, with $u_{i}:= d_{i2}$ and $v_{j}:=d_{2j}-d_{22}$.
Since $ca_{ij}-b_{ij}=u_{i}+v_{j}$, we can define  $c_{ij}=\frac{b_{ij}+u_i}{c}=a_{ij}-\frac{1}{c}v_j$.
Hence, $(A,B)\sim (C,cC)$.
\end{proof}

In case (I) it suffices to examine only rescaled zero-sum games, as best-response dynamics cannot cycle in rescaled coordination games given they are simply weighted potential games.\footnote{In a rescaled coordination game, the ratio of the utilities of the agents is a constant that is independent of the action profile. Hence, any payoff-improving best-response for one agent is also payoff-improving for the other implying that any best response path must terminate at a Nash equilibrium.} 
Similarly, since in case (III) the game has three isolated equilibria, it cannot be equivalent to a rescaled zero-sum game since such games have the same equilibrium set as some zero-sum game and thus their set of equilibria must form a convex polytope\footnote{Any rescaled zero-sum game has the same Nash equilibrium set as some zero-sum game. The set of Nash equilibria in each zero-sum game is a convex polytope~\citep{neumann}.}.
Thus, in case (III), it suffices to explore only rescaled coordination games.

\subsection{Case I: No pure Nash equilibrium (Rescaled Zero-Sum Games)}

We will start off our analysis with Case I, involving rescaled zero-sum games with interior Nash equilibria. The analysis of strong swap regret minimization for RD will involve three steps. Firstly, we will show that all interior trajectories in this case are cycles. This topological argument hinges upon recent topological characterizations of RD trajectories in zero-sum games and variants thereof~\citep{piliouras2014optimization,PiliourasAAMAS2014,mertikopoulos2018cycles,boone2019darwin,nagarajan2020chaos}. These works typically argue for weaker notions of divergence (e.g., Poincar\'{e} recurrence) for almost all initial conditions, whereas in our case we will need to prove periodicity of all interior trajectories. The closest result to our own is in~\citet{boone2019darwin}, where periodicity of RD trajectories is argued for a different class of single-agent, evolutionary zero-sum games with three strategies. Our proof strategy will tailor those arguments to address our different class of games. Armed with this topological result, we will next prove a strong time-average property of $2\times2$ games with a unique interior Nash equilibrium and periodic RD trajectories. Specifically, we will show that along any such trajectory if we fix a specific mixed strategy of one of the agents and compute the time-average of their opponent conditioned only on those points where the first player applies the specified mixed strategy then the time-average of second agent will converge to the mixed Nash equilibrium. This is a significant strengthening of a well-known result that specifies that the time-average of external-regret minimization algorithms in zero-sum games \textit{over their whole trajectory} converges to the marginal of a Nash equilibrium. Leveraging this strong property, it will easily follow that no agent can profitably deviate from their current play even if we allow them to use deviations which are conditional to their current mixed strategy, implying strong swap regret minimization in this case.


\subsubsection{Periodicity of Interior Orbits}
\label{sec:periodic}

We will start of with a well-known theorem that constrains the possible limit behavior of smooth dynamical systems on a plane. Particularly, the Poincar\'{e}-Bendixson theorem provides sufficient conditions under which the limit behavior of such systems is periodic. Let $z=(x_1, y_1)$ denote the state of RD in a $2\times2$ rescaled zero-sum game
where we remove the constrained variables $x_2, y_2$ to express that the system has two degrees of freedom. 
Formally, a limit set $\omega(z)$ of an initial condition $z$ is the union of the limits of all its convergent sub-sequences, i.e., all its possible limit behaviors. Specifically, let $\phi : [0,1]^2 \times \mathbb{R} \rightarrow [0,1]^2$ be the flow of 
RD  such that for any point $z \in  [0,1]^2$, 
$\phi(z, -)$ defines a function of time corresponding to the trajectory of $x$. The \emph{$\omega$-limit (set) of $z$} is formally defined as the set of points $w \in  [0,1]^2$ such that
there exists a sequence $(t_n)$ diverging to $+\infty$ such that 
$\phi(z, t_n) \rightarrow w$.

\begin{theorem}[Poincar\'{e}-Bendixson theorem]
A limit set of a $C^1$ dynamical system over the plane, if non-empty and compact, that does not contain a rest point is a periodic orbit.   
\end{theorem}

We will apply the Poincar\'{e}-Bendixson theorem to show that all interior RD trajectories are periodic for these games.

\begin{theorem}
Given any rescaled zero-sum game with a unique interior Nash equilibrium, all interior trajectories are periodic.    
\end{theorem}

\begin{proof}
First, from~\citet{PiliourasAAMAS2014}, we have that the weighted sum of the Kullback-Leibler (KL) divergences of each agent's current mixed strategy from their respective interior Nash equilibrium strategy is a constant of the motion, i.e., it remains invariant over time for all orbits.
The weights are strictly positive numbers such that, after multiplying these numbers with the current agent utilities, the rescaled game is zero-sum. The KL divergence is equal to zero if and only if the respective probability distributions are equal to each other. 
Thus, the fact that for all non-trivial (i.e. non-equilibrium) orbits, the original weighted sum of KL divergences is positive implies that any such trajectory will stay bounded away from the interior equilibrium. The only other fixed points for RD lie on the boundary and correspond to pure strategy profiles, however, since the Nash equilibrium is interior the KL divergence becomes unbounded as we approach the boundary and thus due to the invariance and non-negativity of KL divergences any interior trajectory will stay bounded away from the boundary. Due to the Poincar\'{e}-Bendixson theorem, we have than any non-trivial orbit of RD will stay bounded away from all fixed points and thus the resulting limit set will be a periodic orbit.  
Finally, it remains to show that not only are the limit set of all trajectories cycles, but all trajectories themselves are cycles.  

  We know that for any interior initial condition $z$, $\omega(z)$ is a periodic orbit and thus a 
  a Jordan curve of the plane $[0,1]^2$, so it is separated  into two connected components: an interior $I$ and an exterior $E$. 
  
  By assumption, we know that there exists an interior rest point of RD, the interior Nash equilibrium. 
  Since RD is no-regret, the time-average of the strategy over $\omega(z)$
  is an equilibrium, that will lie in the convex hull, hence interior to 
  $[0,1]^2$. Let us denote it $z^\star$. 
  We claim that this $z^\star$ has to be a point of $I$. 
  Let us assume not, i.e., $z^\star\in E$. 
  Since $\omega(z)$ is non-trivial, its interior $I$ is non-empty. Pick a point $a \in I$. Draw a semi-infinite ray starting from $z^\star \in E$ in the direction of $a\in I$. 
  Since $I$ is bounded, this ray will transit from $E$ to $I$ then $I$ to $E$ at least once. 
  Hence, it crosses $\omega(z)$ at least two times. 
  Let us name those points $\omega_1$ then $\omega_2$. 
  The sum of KL divergences is a strictly convex function, which attains its global minimum at 
  $z^\star$, so will strictly increase as we advance along the ray. 
  Thus, the weighted sum of KL divergence at $\omega_2$ will be strictly larger than its value at $\omega_1$ , which contradicts the fact that it is a continuous invariant of motion. Therefore, we have that $z^\star \in I$.

  Summarizing, $\omega(z)$ is a cycle, and its time-average (the interior Nash equilibrium $z^\star$) lies in the interior $I$ of $\omega(z)$. 
  We now wish to show that the orbit starting from $z$ is a periodic orbit.
  To prove that, we will argue that $z \in \omega(z)$.
  From $z^\star$, draw a semi-infinite line $L$ in the direction of $z$. 
  As $I$ is bounded, this semi-infinite line have to cross $\omega(z)$ in at least one point. 
  Pick one of these points and call it $z_\omega$. 
  We claim that the weighted sum of KL divergences is equal to the sum of weighted KL divergences between $z^\star$ and $z$ on $L$ only at $z_\omega$. 
  Indeed, the sum of KL divergences is a strict convex function with global minimum at $z^\star$ (equal to zero), so it is strictly increasing as one moves along the line $L$. 
  Thus, $z_\omega$ is the unique intersection point between $L$ and the level set of the weighted KL divergences with value  equal to the weighted KL divergence between $z^\star$ and $z$.
  Therefore, we derive that $z = z_\omega \in \omega(z)$.
\end{proof}

Next, we will move forward to prove strong time-average properties for periodic trajectories of RD dynamics in $2\times2$ games with unique interior Nash equilibrium, which naturally apply to our case of rescaled zero-sum games with unique interior Nash equilibrium.

\subsubsection{Average Opponent Strategies}

\begin{figure}[t]
    \centering
    \includegraphics[width=\textwidth]{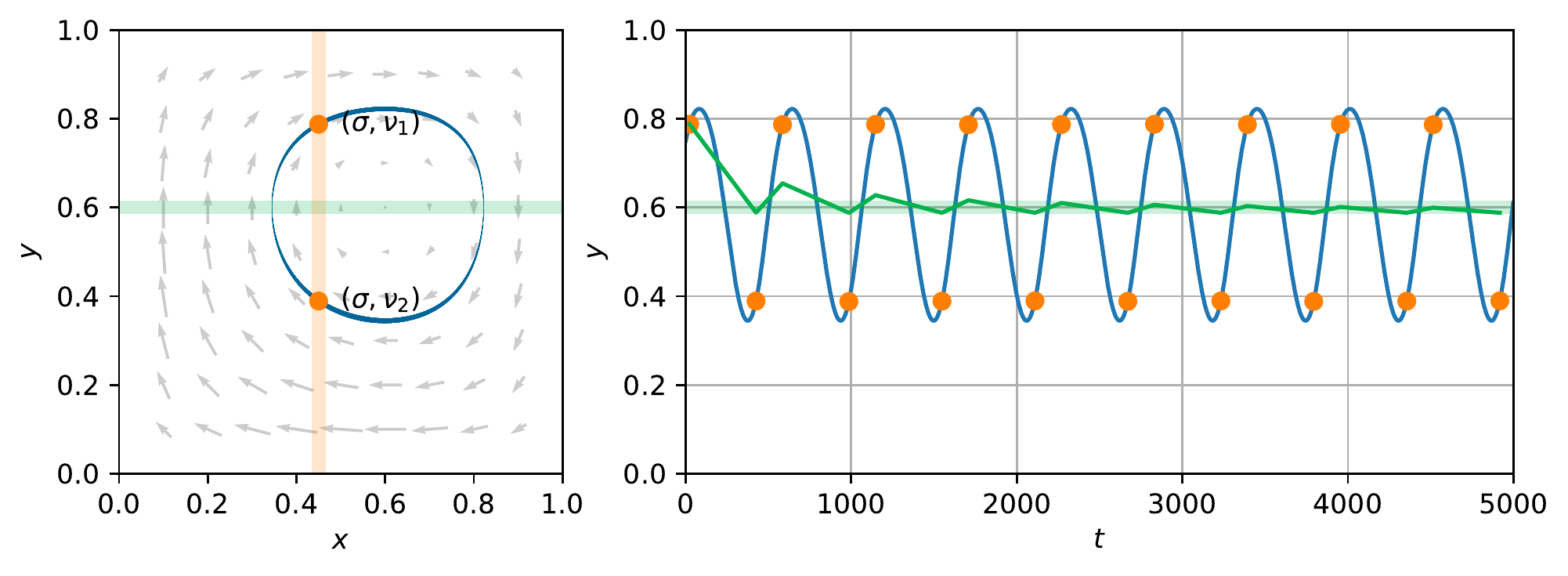}%
    \caption{Illustration of RD for a two-player 2-strategy game.
    The left panel illustrates the player dynamics (with a particular periodic trajectory traced in blue and the corresponding line thickness indicating the row player's update speed, $v_x$).
    For a fixed orbit, whenever the row player passes through mixed strategy $\sigma \in [0,1]$, it observes one of two possible mixed strategies for the column player: $\nu_1, \nu_2 \in [0,1]$.
    The right panel illustrates the recurrence of these orbital crossings, with green series indicating the time-average convergence of the observed opponent strategies to their Nash distribution.
    }
    \label{fig:mp_example}
\end{figure}

We first consider the average opponent strategy that an RD learner observes whenever playing a particular mixed strategy, a concept used to establish our subsequent regret-minimization proofs. 
In order to better understand the notion of `time-average opponent strategy' derived here, consider the example in \cref{fig:mp_example}. 
Here, the left panel illustrates RD trajectories for a two-player 2-strategy game.
The opponent strategy that the row player observes whenever their own strategy, $x$, is fixed to a particular value, $\sigma$, is illustrated via the highlighted vertical band in the figure;
for a fixed orbit, whenever the row player passes through the mixed strategy $x=\sigma$, they observe one of two possible strategies for the (opposing) column player.
The right panel illustrates the time-average of these observed opponent strategies, for which we seek to derive an analytic formula here. 

Below, we derive the time-average opponent strategy result specifically for periodic $2 \times 2$ games (i.e., those wherein the dynamics cycle around an interior Nash equilibrium), as strong swap regret minimization can be handled more straightforwardly for other instances. 

\begin{proposition}\label{prop:rd_avg_strat}
     Consider a periodic $2\times 2$ game with unique interior Nash equilibrium, i.e., a $2\times 2$ game such that all interior orbits of RD are periodic.
    Let $(x^t, y^t)_{t \geq 0}$ be the RD trajectories with some fixed interior initial condition $(x^0, y^0) \in [0,1]^2$. 
    Let $\sigma \in [0,1]$ denote a fixed mixed strategy for the row player, and let $\nu_1, \nu_2 \in [0,1]$ denote the (at most) two possible mixed strategies played by the column player, i.e., such that both $(\sigma, \nu_1)$ and $(\sigma, \nu_2)$ lie on the orbit of RD.\footnote{On top of periodicity of all interior orbits, this proposition assumes that given any such orbit and for any fixed mixed strategy of the first (resp., second) player, the points of the orbit where the first (resp., second) agent applies that mixed strategy correspond to at most two mixed strategies of the second player.
    In \cref{sec:periodic}, we have shown that the necessary conditions for this theorem are satisfied by all rescaled zero-sum games with a unique interior Nash equilibrium.}
    Let $(v^i_x, v^i_y)$ denote the gradient of RD at point $(\sigma, \nu_i)$, for $i=1,2$. 
    Then the time-average opponent strategy observed by the row player up to time $t$ is given by
    \begin{align}\label{eq:rd_avg_strat}
        \frac{\frac{1}{v^1_x}\nu_1 + \frac{1}{v^2_x}\nu_2}{\frac{1}{v^1_x} + \frac{1}{v^2_x}}  + O\left(\frac{1}{t}\right)\, .
    \end{align}
\end{proposition}

\begin{proof}
    We derive the time-average opponent strategy observed within time spans that are multiples of the period of the dynamics, $T$.
    It is sufficient to consider this case, as other cases can be dealt with by noting they modify the time-average observed strategy by at most a constant.\footnote{Specifically, writing $t = NT + r$, with $N \in \mathbb{N}_0$ and $r \in [0, T)$, the contribution of opponent strategies from the interval $(t-r, t)$ to the time-average up to time $t$ goes as $O(\frac{1}{t})$.}
    
    Let $\varepsilon > 0$.
    The average strategy played by the column player whilst the row player plays a strategy within $\varepsilon$ of $\sigma$ is
    \begin{align}\label{eq:avg_strat_def}
        \frac{\int_0^T y^t \mathbbm{1}[x^t \in \left[\sigma-\varepsilon,\sigma+\varepsilon]\right] \mathrm{d}t}{\int_0^T \mathbbm{1}[x^t \in \left[\sigma-\varepsilon,\sigma+\varepsilon]\right] \mathrm{d}t} \, .
    \end{align}
    We aim to derive a limiting expression as $\varepsilon \rightarrow 0$. Let $t_1, t_2 \in (0, T)$ be such that $(x^{t_1}, y^{t_1}) = (\sigma, \nu_1)$ and $(x^{t_2}, y^{t_2}) = (\sigma, \nu_2)$. 
    Since the path of RD is differentiable, we have
    \begin{align*}
        x^{t_1 + t} = \sigma + t v^1_x + O(t)\, , \quad & y^{t_1 + t} = \nu_1 + t v^1_y + O(t) \\
        x^{t_2 + t} = \sigma + t v^2_x + O(t)\, , \quad & y^{t_2 + t} = \nu_2 + t v^2_y + O(t) \, ,
    \end{align*}
    where $(v^i_x, v^i_y)$ denotes the gradient of RD at $(\sigma, \nu_i)$, for $i=1,2$.
    Then note that the condition $x^t \in [\sigma - \varepsilon, \sigma + \varepsilon]$ is satisfied for $t \in [t_1 - \frac{\varepsilon}{v^1_x} + O(\varepsilon), t_1 + \frac{\varepsilon}{v^1_x} + O(\varepsilon)] \cup [t_2 - \frac{\varepsilon}{v^2_x} + O(\varepsilon), t_2 + \frac{\varepsilon}{v^2_x} + O(\varepsilon)]$. Hence, the denominator integral in \cref{eq:avg_strat_def} simplifies to
    \begin{align*}
        2\frac{\varepsilon}{v^1_x} + 2\frac{\varepsilon}{v^2_x} + O(\varepsilon) \, .
    \end{align*}
    As for the numerator in \cref{eq:avg_strat_def}, observe that
    \begin{align*}
        \int_0^T y^t \mathbbm{1}[x^t \in \left[\sigma-\varepsilon,\sigma+\varepsilon]\right] \mathrm{d}t
        = & \sum_{i=1}^2 \int_{t_i - \frac{\varepsilon}{v^i_x} +O(\varepsilon)}^{t_i + \frac{\varepsilon}{v^i_x}  + O(\varepsilon)} y^t \mathrm{d}t \\
        = & \sum_{i=1}^2 \int_{-\frac{\varepsilon}{v^i_x} + O(\varepsilon)}^{\frac{\varepsilon}{v^i_x}  + O(\varepsilon)} \left(\nu_i + t v^i_y + O(t)\right) \mathrm{d}t \\
        = & \sum_{i=1}^2 \left(\left(2\frac{\varepsilon}{v^i_x} + O(\varepsilon)\right) \nu_i + O(\varepsilon) \right) \, .
    \end{align*}
    Altogether, the average strategy faced by the row player in \cref{eq:avg_strat_def} is rewritten as
    \begin{align*}
        \frac{\sum_{i=1}^2 \left(2\frac{\varepsilon}{v^i_x} + O(\varepsilon)\right) \nu_i + O(\varepsilon)}{2\frac{\varepsilon}{v^1_x} + 2\frac{\varepsilon}{v^2_x} + O(\varepsilon)} \, .
    \end{align*}
    Letting $\varepsilon \rightarrow 0$, we compute
    \begin{align*}
        \lim_{\varepsilon \rightarrow 0} \frac{\sum_{i=1}^2 \left((2\frac{\varepsilon}{v^i_x} + O(\varepsilon)\right) \nu_i + O(\varepsilon)}{2\frac{\varepsilon}{v^1_x} + 2\frac{\varepsilon}{v^2_x} + O(\varepsilon)}
        = & \lim_{\varepsilon \rightarrow 0} \left\lbrack \frac{(2\frac{\varepsilon}{v^1_x} \nu_1 + 2\frac{\varepsilon}{v^2_x} \nu_2)}{2\frac{\varepsilon}{v^1_x} + 2\frac{\varepsilon}{v^2_x} + O(\varepsilon)} + \frac{O(\varepsilon)}{2\frac{\varepsilon}{v^1_x} + 2\frac{\varepsilon}{v^2_x} + O(\varepsilon)} \right\rbrack \\
        = & \frac{\frac{1}{v^1_x} \nu_1 + \frac{1}{v^2_x} \nu_2}{\frac{1}{v^1_x} + \frac{1}{v^2_x}} \, .
    \end{align*}
\end{proof}

\begin{remark}\label{rem:param}
    In the representation \cref{eq:rd_avg_strat} of the time-average opponent strategy induced by \cref{prop:rd_avg_strat}, the weights  do not depend on the underlying parametrization choice.
    Namely, consider any reparameterization of mixed strategy $(x,y)$ via $(\theta^1,\theta^2) = (g(x),g(y))$, where $g$ is a non-decreasing  differentiable  function. Then, a similar line of arguments leads to the following result:
    \begin{align*}
    \frac{\frac{1}{v^1_\theta}\; g(\nu_1) + \frac{1}{v^2_\theta} \; g(\nu_2)}{\frac{1}{v^1_\theta} + \frac{1}{v^2_\theta}} + O\left(\frac{1}{t}\right)\;.
    \end{align*}
\end{remark}
\cref{prop:rd_avg_strat} implies in particular that, whenever the $2\times 2$ game of interest is also zero-sum while exhibiting a unique interior Nash equilibrium, the time-average opponent strategy identifies to the Nash equilibrium itself.
This result is proven in the following proposition. 
\begin{proposition}\label{lemma:Nash_Opponent}
    Consider a periodic $2\times 2$ game with unique interior Nash equilibrium, i.e., a $2\times 2$ game such that all interior orbits of RD are periodic. Furthermore, assume that given any such periodic orbit and for any fixed mixed strategy of the first (resp., second) player, the points of the orbit where the first (resp., second) agent applies that mixed strategy correspond to at most two mixed strategies of the second player such that the equilibrium strategy of the second (resp., first) player lies in their convex combination.\footnote{In \cref{sec:periodic}, we have shown that the necessary conditions for this theorem are satisfied by all rescaled zero-sum games with a unique interior Nash equilibrium.} Then,
    for any fixed first player strategy, the corresponding time-average opponent strategy in RD is the Nash strategy of the opponent. 
\end{proposition}

\begin{proof}
Consider the logit parameterization of strategies in these 2-strategy games, wherein a player's strategy is described by a single real number $\tau \in \mathbb{R}$, and the corresponding strategy probabilities are given by $e^{\tau}(1+e^\tau)^{-1}$ and $(1+e^\tau)^{-1}$.
Given two joint strategies $(\sigma, \tau_1)$, $(\sigma, \tau_2)$ in logit parameterization on the orbit of RD, the gradient for the row player under RD is given by the difference in payoffs for row strategies 1 and 2 against the column player's joint strategy, namely:
\begin{align*}
    \begin{pmatrix}
        1 & -1 
    \end{pmatrix}
    A
    \begin{pmatrix}
        e^{\tau_2}(1+e^{\tau_2})^{-1} \\
        (1+e^{\tau_2})^{-1} \\
    \end{pmatrix} \, ,
\end{align*}
and similarly for the column player. 
For shorthand, denote  $z := (1,\ -1)$ and $y(\tau) := (e^\tau(1+e^\tau)^{-1} ,\ \ (1+e^\tau)^{-1})^\top$. 
Per \cref{prop:rd_avg_strat} and \cref{rem:param}, the average opponent strategy that the row player experiences when playing $\sigma$ is expressed by
\begin{align*}
    \frac{\frac{e^{\tau_1}(1+e^{\tau_1})^{-1}}{|zAy(\tau_1)|} + \frac{e^{\tau_2}(1+e^{\tau_2})^{-1}}{|zAy(\tau_2)|}}{\frac{1}{|zAy(\tau_1)|} + \frac{1}{|zAy(\tau_2)|}} \, .
\end{align*}

Let $p_i:=e^{\tau_i}(1+e^{\tau_i})^{-1}$ so that $y(\tau_i)=(p_i,1-p_i)^\top$, for $i=1,2$.
Similarly, let $p_*$ denote the probability of the first strategy in the unique Nash equilibrium of the game, so that $y_* := (p_*, 1-p_*)^{\top}$. For the statement of the proposition to be valid, it remains to prove the following relation:
\begin{align*}
    \frac{\frac{1}{|zAy(\tau_1)|} p_1 + \frac{1}{|zAy(\tau_2)|} p_2}{\frac{1}{|zAy(\tau_1)|} + \frac{1}{|zAy(\tau_2)|}} \;=\; p_* \, ,
\end{align*}
which rewrites equivalently as
\begin{align}\label{eq:relation}
    {\frac{1}{|zAy(\tau_1)|} (p_1-p_*) + \frac{1}{|zAy(\tau_2)|} (p_2-p_*)} \;=\; 0 \, .
\end{align}
By definition of the Nash equilibrium $y_*$, recall that $zAy_*=0$, so that 
$$
|zAy(\tau_i)| =  |zAy(\tau_i)-zAy_*| = |zAz^\top| \times |p_i-p^*|, \quad \mbox{for }i=1,2\;.
$$
This allows one to rewrite the targeted relation \cref{eq:relation} as
\begin{align*}
    \frac{1}{|zAz^\top|} \left( \frac{(p_1-p_*)}{|p_1-p_*|} +  \frac{(p_2-p_*)}{|p_2-p_*|} \right) \;=\; 0 \, ,
\end{align*}
which holds true since $p_* \in [p_1, p_2]$ as $(\sigma,\tau_1)$ and $(\sigma,\tau_2)$ belong to the same orbit of RD. 
\end{proof}

We next prove the result establishing minimization of strong swap regret for Case I. 
\begin{proposition}\label{prop:interior}
    In any generic $2\times2$ game with a unique interior Nash equilibrium, RD minimizes strong swap regret. 
\end{proposition}

\begin{proof}
    By \cref{lemma:equi} along with the fact that the game does not have any pure Nash equilibrium, we have that the game has to be equivalent to a rescaled zero-sum game. 
    In such a case, we have that all interior trajectories of RD are periodic. 
    Following the line of arguments of \cref{prop:rd_avg_strat} and \cref{lemma:Nash_Opponent}, we can conclude that RD minimizes strong swap regret, since for any fixed first player strategy, the average opponent strategy it observes is the Nash equilibrium component of the opponent and thus no profitable deviation is possible even if the agent is allowed to condition their deviating strategy to their current mixed strategy.
\end{proof}

The argument about the time-average convergence to the Nash equilibrium is specific to zero-sum games with a unique interior Nash equilibrium. 
In some sense, this is the most interesting subclass in which to study strong swap regret minimization, since this minimization result is itself strongest in this class of games; we will see below that in other cases, strong swap regret minimization coincides with weaker notions of regret.

\subsection{Case II: Games with a Unique Pure Nash Equilibrium}
In the case with a unique Nash equilibrium on the simplex boundary, as argued above, RD converges to the pure Nash equilibrium, which we will exploit to establish a link between external regret and strong swap regret. 
\begin{proposition}
    In any generic $2\times2$ game with a unique pure Nash equilibrium, RD minimizes strong swap regret. 
\end{proposition}

\begin{proof}
In the case of generic $2\times 2$ games with a unique pure Nash equilibrium, the game is strictly dominance solvable. Namely, after the iterated elimination of pure strategies that are strictly dominated by pure strategies, we are left with a single strategy for each agent. 
Indeed, the best response dynamic has a unique sink and for at least one of the two agents one strategy strictly dominates the other.
In such games, RD (as well as a wide range of other evolutionary dynamics) is bound to converge to playing the (iteratively) undominated strategies with probability one~\citep[Proposition 4.6]{Weibull}. However, in any small enough neighborhood around this strict Nash equilibrium, the best response of each agent is unique and so is their Nash equilibrium strategy. Hence, given any interior initial condition, after a finite amount of time, the optimal deviation strategies for strong swap regret are going to be independent of the (exact) mixed strategy of the agent and, instead, they will stay fixed and equal to the Nash equilibrium strategy of that agent. Thus, strong swap regret will be within a constant of the external regret of RD, which is known to be bounded in general games. Thus, RD also minimizes strong swap regret. 
\end{proof}

\subsection{Case III: Games with Two Pure Nash Equilibria (Rescaled Coordination Games)} 
This is the last case of $2\times 2$ games, where we examine games with multiple Nash equilibria. Given the fact that some trajectories converge to pure Nash equilibria and some to mixed Nash equilibria, we need specialized arguments.  

\begin{lemma}
    In any $2\times2$ games with several Nash equilibria on the simplex boundary, RD minimizes strong swap regret.
\end{lemma}

\begin{proof}
    In this case, the game is equivalent to a rescaled coordination game per \cref{sec:equiv_rescaled_games}. 
    However, rescaled coordination games are weighted potential games where the potential function of the player is the utility of, e.g., the first player. 
    In this case RD is known to converge to Nash equilibria for all interior initial conditions~\citep{papadimitriou2018nash}.
    In terms of the possible limit points of different trajectories, we have two possibilities: either the limit is a pure Nash equilibrium, or it is an interior (mixed) Nash equilibrium. 
    
    In the case where the trajectory converges to a pure Nash equilibrium, the boundedness of strong swap regret follows from the same argument as in the case of strictly dominance solvable games converging to a pure Nash equilibrium. 
    Thus, the only remaining case is that of trajectories converging to a mixed Nash equilibrium. 
    
    In the case of rescaled coordination games with interior Nash equilibrium, there always exists a positive weight such that the weighted difference of KL divergences between each agent's mixed Nash equilibrium and their evolving state remains invariant over time (see~\citealt{Hofbauer96,nagarajan2018three}). 
    Specifically, given any pair of (interior) mixed strategies $x=(x_1, x_2)$ and $y=(y_1, y_2)$, there exists $w>0$ such that the KL divergence $w D_\text{KL}\big(p \parallel x^t\big) - D_\text{KL}\big(q \parallel y^t\big)  = w \sum_i p_i\ln(\frac{p_i}{x_i^t}) - \sum_j q_i\ln(\frac{q_j}{y_j^t})$ between the Nash equilibrium  $(p, q)$ and the state $(x^t, y^t)$ of the system is invariant over time.
    Clearly, the value of this invariant function evaluated at the mixed Nash equilibrium $(p,q)$ is equal to zero and by continuity the value of this function must be equal to zero for any trajectory that converges to that mixed Nash equilibrium. Finally, the important property implied by this invariant function is that any trajectory
    that converges to the mixed Nash equilibrium $(p,q)$ cannot intersect the lines $x_1=p_1$, $y_1=q_1$. 
    If that was the case and exactly one of the two conditions $x_1=p_1$, $y_1=q_1$ was satisfied, then exactly one of the two KL divergence terms would be equal to zero, implying that the invariant function cannot be equal to zero. 
    On the other hand, if the trajectory intersects the lines $x_1=p_1$, $y_1=q_1$ simultaneously then by uniqueness of the solution of the ODE this is the trivial equilibrium trajectory that does not converge to equilibrium since it already starts at it. 
    Hence, when encoding the state of the $2\times 2$ replicator dynamics in the $[0,1]^2$ plane expressing the variables $x_1,y_1$, then any trajectory converging to the interior Nash equilibrium must be strictly contained in one of the four parallelograms encoded by the lines $x_1=p_1$, $y_1=q_1$. Thus, the best response of both agents along any such trajectory remains fixed over time and, as a result, the maximal strong swap regret is achieved by deviating to that fixed best response throughout the system orbit. 
    Therefore, strong swap regret is upper bounded by standard external regret, which is well-known to be bounded for RD (e.g., see~\citealt{mertikopoulos2018cycles}).
\end{proof}

\section{Linking Strong Swap Regret to Other Optimization Notions}\label{sec:linking_swap_other}

Having established that RD minimizes strong swap regret in generic $2 \times 2$ games, we now present relationships to more standard perspectives in online learning. 

\subsection{Strong Swap Regret versus External Regret Minimization}

A natural question to ask is whether minimization of strong swap regret follows immediately from minimization of standard (external) regret in the special case of $2\times2$ games. 
The following theorem shows that this is not the case.
\begin{theorem}
    There exist algorithms that minimize external regret, which exhibit unbounded strong swap regret in some $2\times2$ games.
\end{theorem}

\begin{proof}
    Consider an arbitrary tuple of algorithms that minimize external regret (e.g., bounded external regret algorithms such as RD), with one such algorithm used by each of the agents. 
    Given such a tuple of algorithms, we define a slight variant of them that initializes them as playing a pre-computed CCE in a chosen target game that does not have bounded strong swap regret. 
    For example, consider the $2\times2$ game where both agents receive a payoff of $1$ if they choose the same strategy (i.e., both play strategy $a$ or $b$), otherwise receiving a payoff of $0$. 
    Now, consider the following strategy distribution in this game: 
    with probability $\nicefrac{1}{2}$, both agents play the mixed strategy where they choose the first strategy with probability $\nicefrac{3}{4}$; 
    with probability $\nicefrac{1}{2}$, both agents play the mixed strategy where they choose the second strategy with probability $\nicefrac{3}{4}$.
    This overall strategy distribution is a CCE.\footnote{The overall distribution chooses pure strategy profiles $(a,a), (b,b)$ each with probability $\nicefrac{5}{16}$ and the profiles $(a,b), (b,a)$ with probability $\nicefrac{3}{16}$; it is easy to verify this distribution satisfies the CCE definition.} 
    Thus, dividing the time interval into epochs of length one, we can have the agents reproduce this CCE by alternating between playing the two described mixed strategies.
    
    The agents continue producing this predefined play as long as no deviation from it is observed.
    If any deviation is observed from any individual agent, then all agents subsequently deviate from this play, from this point onward always applying their assigned regret-minimizing algorithm (e.g., RD). 
    Clearly, each of these algorithms still minimizes external regret, as even if a deviation is observed leading to this switch away from the original CCE mixed strategy, the total regret experienced on any single period is bounded and the total regret of the individual algorithm used (e.g., RD) is bounded as well.
    On the other hand, by construction the resulting play does not minimize strong swap regret, as any of the two mixed strategies applied by the agents is not a best response during the time epoch that is being played;
    thus, instead the agents would rather deviate from playing their assigned strategy, not with probability $\nicefrac{3}{4}$, but with probability $1$.   
\end{proof}

\subsection{Strong Swap Regret in the Case of Online Learning With Two Actions}

While the definition of strong swap regret was motivated by the need for more granular techniques in the study of learning  in games, a related question is whether this notion also yields insights for more traditional online learning settings.
Below, we show that it is possible to provide positive regret results even if we allow for the deviating strategies to depend on the current mixed strategy, the telltale attribute of strong swap regret, under specific conditions on the learning trajectory. 
Specifically, we explore the conditions under which we can prove positive results for strong swap regret, not in the case of $2\times2$ games, but rather in the case of online learning with two actions (i.e., strategies) whose utilities evolve continuously over time. 

\begin{theorem}
    Given bounded and continuous utility functions $u^t_1, u^t_2$, the aggregated regret of RD remains bounded even if we consider only the part of the history of play where replicator chooses the first action with probability $p \in (p_0-\epsilon, p_0+\epsilon)$ for any $p_0\in (0,1)$, as long as $\epsilon$ is small enough such that the replicator trajectory always exits the interval from the opposite end of the interval from the one it enters it. 
\end{theorem}

\begin{proof}
Let us create a possibly infinite sequence of time intervals $(a_j,b_j)$ whose union captures the set of time instances $t$ such that $p^t\in (p_0-\epsilon, p_0+\epsilon)$.  To explore the total regret for a player in online learning in the time intervals where the probability of playing their first strategy is between 
$p_0-\epsilon$ and $p_0+\epsilon$, it suffices to track the total regret for deviating against any fixed strategy for the union of $(a_j,b_j)$ intervals. 

In the case of RD in any such interval $(a_j,b_j)$, the regret for not deviating to some fixed strategy $i$ is given as follows:
\begin{align*}
\dot{p}_i = p_i(u_i-\sum_i p_i u_i) \, .
\end{align*}
Rearranging,
\begin{align*}
\frac{\dot{p}_i}{p_i} &= u_i-\sum_i p_i u_i\\
\therefore \ln(p_i(b_j))-\ln(p_i(a_j)) &= \int^{b_j}_{a_j} u_i^{\tau}d\tau-\int^{b_j}_{a_j} \sum_i p_i^{\tau} u_i^{\tau} d\tau \, .
\end{align*}
Without loss of generality, we can assume that the trajectory enters the $(p_0-\epsilon, p_0+\epsilon)$ interval from below. In this case, 
the regret in the union of these times intervals is 
\begin{align*}
    \sum_j \ln(p_i(b_j))-\ln(p_i(a_j)) &= \ln(p_0+\epsilon)-\ln(p_0-\epsilon) + \ln(p_0-\epsilon) - \ln(p_0+\epsilon) + \dots \\
    &\leq  \ln(p_0+\epsilon)-\ln(p_0-\epsilon)
\end{align*}
and thus remains bounded for all time.
\end{proof}

Overall, these insights indicate that there exist interesting positive and negative results linking strong swap regret to traditional regret notion.

\section{Experiments}\label{sec:experiments}

\begin{figure}
\centering
  \begin{subtable}[t]{\textwidth}
    \centering
    \setlength{\tabcolsep}{3pt}
    \resizebox{\textwidth}{!}{%
        \begin{tabular}{cccccccccccc}
        \toprule
        \textbf{Chicken} &
        \textbf{Battle} &
        \textbf{Hero} &
        \textbf{Compromise} &
        \textbf{Deadlock} &
        \textbf{Prisoner's dilemma} &
        \textbf{Stag hunt} &
        \textbf{Assurance} &
        \textbf{Coordination} &
        \textbf{Peace} &
        \textbf{Harmony} &
        \textbf{Concord} \\
        \textbf{(Ch)} &
        \textbf{(Ba)} &
        \textbf{(Hr)} &
        \textbf{(Cm)} &
        \textbf{(Dl)} &
        \textbf{(Pd)} &
        \textbf{(Sh)} &
        \textbf{(As)} &
        \textbf{(Co)} &
        \textbf{(Pc)} &
        \textbf{(Ha)} &
        \textbf{(Nc)} \\
        \midrule
        $\begin{bmatrix} 2 & 3 \\ 1 & 4 \end{bmatrix}$ &
        $\begin{bmatrix} 3 & 2 \\ 1 & 4 \end{bmatrix}$ &
        $\begin{bmatrix} 3 & 1 \\ 2 & 4 \end{bmatrix}$ &
        $\begin{bmatrix} 2 & 1 \\ 3 & 4 \end{bmatrix}$ &
        $\begin{bmatrix} 1 & 2 \\ 3 & 4 \end{bmatrix}$ &
        $\begin{bmatrix} 1 & 3 \\ 2 & 4 \end{bmatrix}$ &
        $\begin{bmatrix} 1 & 4 \\ 2 & 3 \end{bmatrix}$ &
        $\begin{bmatrix} 1 & 4 \\ 3 & 2 \end{bmatrix}$ &
        $\begin{bmatrix} 2 & 4 \\ 3 & 1 \end{bmatrix}$ &
        $\begin{bmatrix} 3 & 4 \\ 2 & 1 \end{bmatrix}$ &
        $\begin{bmatrix} 3 & 4 \\ 1 & 2 \end{bmatrix}$ &
        $\begin{bmatrix} 2 & 4 \\ 1 & 3 \end{bmatrix}$ \\
        \bottomrule
        \end{tabular}
        }
        \vspace{-5pt}
        \caption{}
        \label{fig:2x2_payoffs_bruns}
    \end{subtable}
    \begin{subfigure}{\textwidth}
        \centering
        \includegraphics[width=0.75\textwidth]{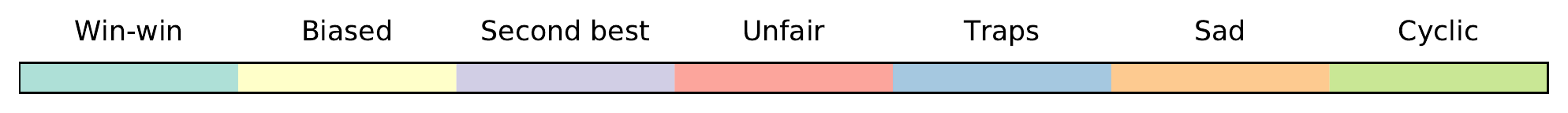}
    \end{subfigure}\\[-5pt]
    \begin{subfigure}{\textwidth}
        \centering
        \includegraphics[width=0.95\textwidth]{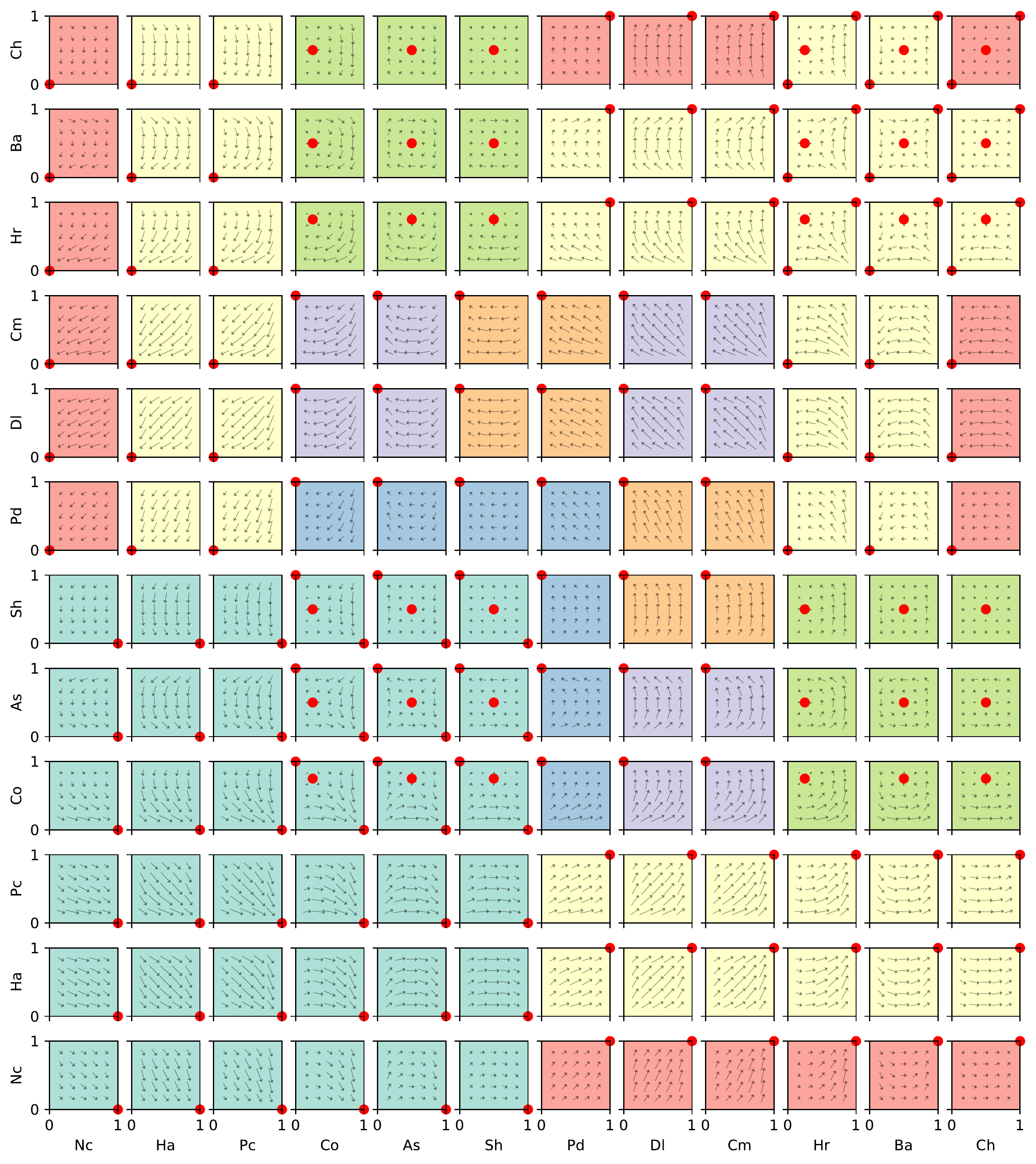}%
        \vspace{-8pt}
        \caption{}
        \label{fig:2x2_rd_bruns}
    \end{subfigure}%
    \vspace{-8pt}
    \caption{
    \subref{fig:2x2_payoffs_bruns} Row player payoffs for $2 \times 2$ games defined by \citet{bruns2015names}. Note: in the convention of \citet{bruns2015names}, corresponding column player payoffs are defined as a transpose along the {anti}-diagonal.
    \subref{fig:2x2_rd_bruns} RD run on 144 games composed of the indicated combinations of row and column player payoffs. 
    Each subplot shows the dynamics for RD on a distinct game, with background colors corresponding to the classes of games originally defined by \citet{bruns2015names}.
    Nash equilibria for each game are indicated by the red points.
    }
\end{figure}

\begin{figure}[t]
    \centering
    \includegraphics[width=\textwidth]{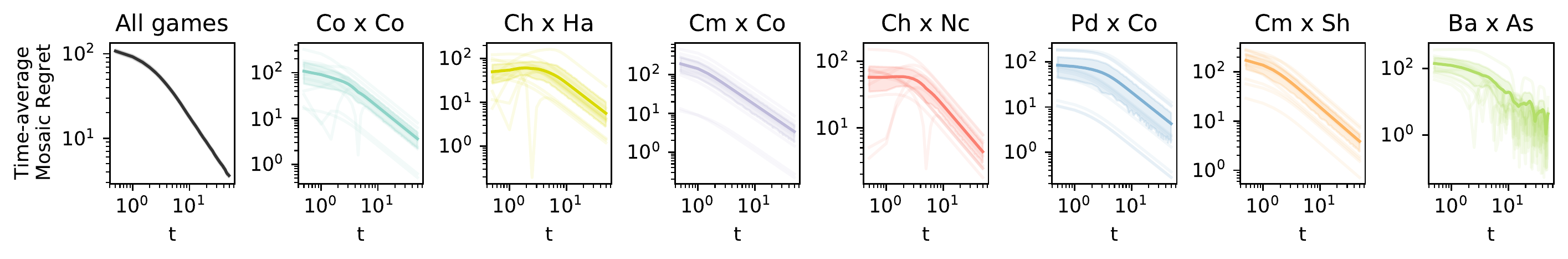}%
    \caption{Mosaic regret convergence for $2 \times 2$ games. For each of the 144 games (shown earlier in \cref{fig:2x2_rd_bruns}), we run 10 trials of RD with independently random initial strategies for both the row and column players, computing mosaic regret through time.
    The lefternmost subplot above summarizes the time-average convergence of mosaic regret, $\nicefrac{1}{T}\sum_{t=0}^{T} MR^t$, across all games and trials. 
    The subsequent subplots show convergence for one example game per Bruns class, with classes corresponding to those in \cref{fig:2x2_rd_bruns}; in each of these latter subplots, each line indicates an individual trial, and the shaded region shows the average mosaic regret and 95\% confidence interval across all trials.}
    \label{fig:2x2_bruns_strongregret_converge}
\end{figure}

In this section, we ground our theoretical findings in experimental analysis of RD in numerous games of varying characteristics.

\subsection{$\Phi$-regret Minimization via RD in $2 \times 2$ games}

As mentioned earlier, practical experiments preclude one from demonstrating that strong swap regret is minimized, as time-discretization implies that any individual mixed strategy observed in a single RD trajectory may not be revisited in future trajectories. 
Thus, we focus our empirical analysis here on mosaic regret, the practically-realizable instance of $\Phi$-regret we introduced in \cref{sec:mosaic_regret}. 

We first conduct a set of experiments illustrating the minimization of mosaic regret under RD in a wide suite of $2 \times 2$ games.
As noted earlier, despite their simple structure, these payoffs form an important class capturing numerous canonical games that have received tremendous attention throughout the game theory literature (e.g., Matching Pennies, Prisoner's Dilemma, Stag Hunt, etc.) \citep{guyer19722,klos2010evolutionary,bruns2015names,robinson2005topology,rapoport1966taxonomy}.
Thus, despite their apparent simplicity, revisiting these under the lens of $\Phi$-regret and mosaic regret reveals that the behavior of learning dynamics in these games is far from well-understood.

The specific games we consider are those defined by \citet{bruns2015names}, which taxonomizes a collection of $2 \times 2$ games into several distinct classes by considering the patterns of payoffs received by each player.
\citet{bruns2015names} identifies 12 sets of basis payoffs corresponding to canonical games of varying characteristics, summarized for the row player in \cref{fig:2x2_payoffs_bruns};
corresponding column player payoffs for these games are defined in \citet{bruns2015names} as the transpose along the \emph{anti}-diagonal, which we also use for consistency.

The combination of these basis payoffs enables the definition of a large collection of 144 $2 \times 2$ games with varying characteristics (e.g., cyclical games, win-win games, pure coordination games, etc.).
For each of these games, we visualize the vector field summarizing the behavior of agents playing RD in \cref{fig:2x2_rd_bruns}.
The $x$- and $y$-axes here, respectively, show the row and column player's probabilities of playing their first strategies.
\Cref{fig:2x2_bruns_strongregret_converge} visualizes an evaluation of the mosaic regret across this diverse set of games.
For simplicity, the partitioning scheme $\Sigma$ we use for computing mosaic regret is to subdivide the first player's mixed strategy space $x$ into 10 discrete, equally-sized bins. 
To generate the mosaic regret results, for each game we run 10 trials of RD (each with an independently-initialized set of seed strategies for the row and column players).
We plot the time-average mosaic regret averaged across all 1440 combinations of games and trials in the lefternmost subplot of \cref{fig:2x2_bruns_strongregret_converge}, verifying our earlier theoretical convergence results for these general classes of games.
Moreover, to understand the the game-specific convergence characteristics, we plot the time-average mosaic regret for one example game per Bruns class (i.e., one per `win-win', `biased', `second best', `unfair', `traps', `sad', and `cyclic' class of games) in the subsequent subplots of \cref{fig:2x2_bruns_strongregret_converge}.
The time-average mosaic regret converges across all these games, with a notable observation being that final game, Ba $\times$ As, suffers from higher-variance across trials, which is explained primarily due to the cyclical nature of this particular game.

\subsection{Mosaic Regret Beyond $2 \times 2$ Games}

\begin{figure}[t!]
    \stepcounter{figure}%
    \setcounter{row}{1}%
    \setcounter{subfigure}{0}%
    \renewcommand{\thesubfigure}{\Roman{row}.\alph{subfigure}}%
    \centering
    \begin{tabularx}{\textwidth}{p{0.0in}YYY}
        & Raw Trajectories & Retained Subset & Average Observed Column Strategies\\
        \gameRow{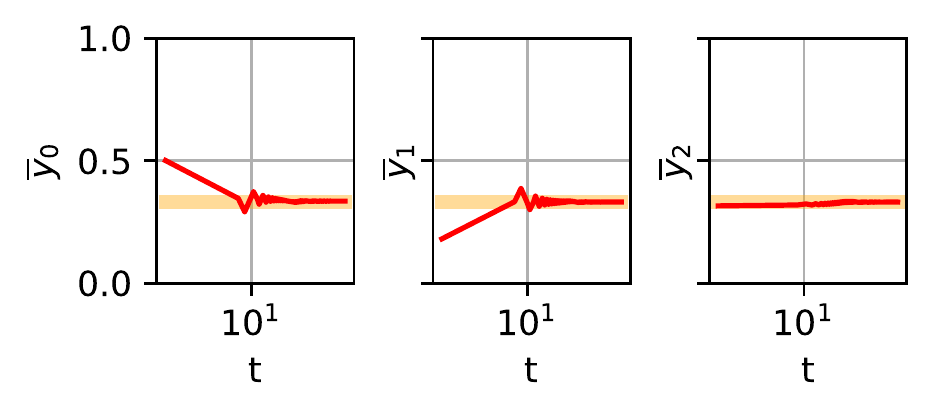}{Game ($A_1$, $B_1$)}{game1}
        \gameRow{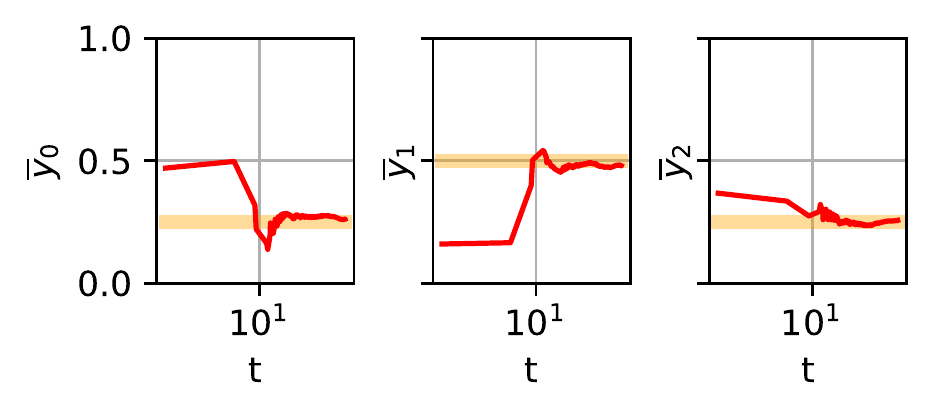}{Game ($A_2$, $B_2$)}{game2}
        \gameRow{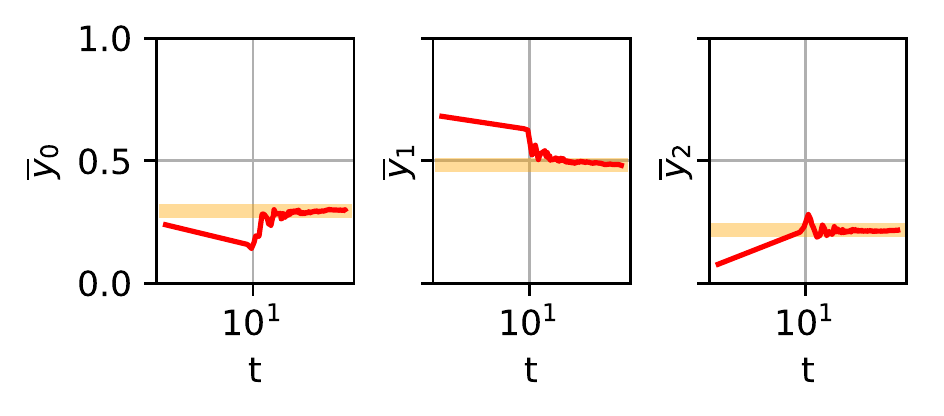}{Game ($A_3$, $B_3$)}{game3}
    \end{tabularx}
    \caption{
        Several instances of $3 \times 3$ games where we observe evidence that RD minimizes mosaic regret.
        In all plots, the Nash equilibrium is indicated in orange for each player.
        Each row corresponds to a zero-sum game (with payoffs described in the main text).
        The first column visualizes points corresponding to raw trajectories exhibited under RD.
        Next, we retain only the partition of points wherein the first player's strategy, $x$, is within a small neighborhood of a reference KL-divergence from its Nash equilibrium; the second column plots the joint strategy under this filtering scheme.
        In the final column, we use all such sampled column player strategies up to the given timestep $t$, $y_{<t}$, to compute the time-average opponent strategy, $\bar{y}^t$.
        The notable observation is that the time-average opponent observed strategy converges to the column player's Nash equilibrium (in orange) in each instance, under the described partitioning scheme.
    }
    \label{fig:3x3_results}
\end{figure}

We next consider several increasingly-complex instances of larger games, where we observe evidence that RD minimizes mosaic regret.
Each row of \cref{fig:3x3_results} visualizes results associated with a distinct $3 \times 3$ game, where the respective row player payoffs are defined as
\begin{align}\label{eq:rps_like_payoffs}
    A_1 =  
    \begin{bmatrix}
    0& -1& 1\\
    1& 0& -1\\
    -1& 1& 0\\
    \end{bmatrix} 
    \qquad
    A_2 =  
    \begin{bmatrix}
    0& -1& 2\\
    1& 0& -1\\
    -2& 1& 0\\
    \end{bmatrix} 
    \qquad
    A_3 =  
    \begin{bmatrix}
    1& -1& 1.2\\
    1& 0& -1\\
    -1& 1& -0.5\\
    \end{bmatrix} \, ,
\end{align}
with corresponding column player payoffs $B_i = -A_i$ for each game $i$.
These payoffs, in order, correspond to the canonical game of Rock--Paper--Scissors (RPS), a variant of RPS with symmetrically-biased payoffs (yielding an non-uniform Nash equilibrium for both players), and an asymmetric variant of RPS (yielding distinct Nash equilibria for either player).

The first column in \cref{fig:3x3_results} visualizes points corresponding to raw trajectories exhibited under RD.
For each game, we subsequently consider a mosaic regret deviation partition as follows: we select an arbitrary reference point in the row player's trajectories, $x_r$, and compute the KL-divergence of the point with respect to the row player's Nash equilibrium, $x_*$ (indicated in orange in the corresponding plots), i.e., $d \triangleq D_{KL}(x_r||x_*)$.
Subsequently, throughout all raw trajectories, we retain only the joint strategies $(x,y)$ wherein the first player's distance to the Nash equilibrium is close to this reference KL-divergence (i.e., $D_{KL}(x||x_*)  \in [d-\varepsilon, d+\varepsilon]$, where $\varepsilon$ is a small threshold parameter).
The set of points that constitute the resulting partition are visualized in the `Retained Subset' column of \cref{fig:3x3_results}.
Next, as in \cref{prop:rd_avg_strat}, for each timestep $t$, we use all such sampled column player strategies up to that timestep, $y_{<t}$, to compute the time-average opponent strategy, $\bar{y}^t$.
The final column of \cref{fig:3x3_results} visualizes this time-average `observed' column strategy for each of the games, with the notable observation being that it converges to the column player's Nash equilibrium (in orange) for each game.
In other words, as in our earlier $2 \times 2$ analysis and experiments, RD minimizes mosaic regret in these larger games, under the partitioning scheme described.

However, this property may not necessarily hold in more general settings.
For example, consider again the game of Rock--Paper--Scissors as defined by $A_1$ in \cref{eq:rps_like_payoffs}, but where both players' strategies are initialized identically (i.e., $x^0=y^0$), thus ensuring that $x^t = y^t$ for all $t$ under RD due to symmetry of the game payoffs.
Subsequently, both players consistently attain the value of the game due to the symmetry in chosen strategies.
In such a scenario, at any time $t$, either player can increase their payoff by best-responding, thus ensuring unbounded mosaic regret. 
Such counterexamples, thus, motivate a need to more thoroughly investigate strong swap regret and mosaic regret convergence in larger and more general classes of games. 

\section{Discussion}\label{sec:discussion}

There has been increasing interest in recent years in defining alternative, generalized notions of regret to better enable stronger theoretical guarantees drawn about learning algorithms~\citep{lehrer2003wide}.
For instance, \citet{mohri2014conditional} introduce \emph{conditional swap regret}, which generalizes swap regret to enable deviations to condition on the learner's history of actions (as opposed to the standard unconditioned deviations). 
\citet{arora2012online} investigate strengthening of the utility classes (resp., `adversary classes' in their case, as they focus on losses), enabling adaptation of the utilities (resp., losses defined by the adversary) on the learner's past actions. 
\citet{morrill2020hindsight} explore existing and defines new forms of strategy deviations in the space of extensive-form games.
\emph{Dynamic regret} \citep{zinkevich2003online}, measures the regret with respect to a sequence of changing baselines, and likewise \emph{adaptive regret} \citep{hazan2007adaptive,daniely2015strongly} computes regret by considering utility differences in windows of time intervals of a specific length.
Overall, the primary motivation behind this line of investigation has been similar to ours: to generalize canonical forms of regret by introducing degrees of freedom in terms of how regret is computed (e.g., by permitting more general classes of deviations, varying the time intervals over which regret is computed, allowing the baseline compared against to change in various ways, etc.).

In our work, we illustrated through a re-examination of RD that, under the lens of $\Phi$-regret, one can study stronger properties associated with the real-time behaviors of such learning processes in games, thus establishing strong convergence guarantees even over the space of mixed strategies. 
To our knowledge, our work provides a first demonstration of the strongest form of $\Phi$-regret being minimized by the well-studied and simple-to-implement RD algorithm in both theory and practice.
While our theoretical and empirical results focused primarily on exhaustively exploring in the space of generic $2 \times 2$ two-player zero-sum games, examination of larger games under this paradigm will make for an interesting avenue of future exploration.
Moreover, another line of future work will involve investigating whether discrete-time dynamical systems related to RD, such as multiplicative weights update, also minimize strong swap or mosaic regret in similar or more general circumstances.

\section{Acknowledgements}
The authors are grateful to Thore Graepel and Marc Lanctot for their feedback during the paper writing process.

\bibliographystyle{plainnat}
\bibliography{references,ms}

\end{document}